\newcommand{\JASA}{0}

\if\JASA1
\documentclass[12pt]{article}
\usepackage{amsmath}
\usepackage{graphicx}
\usepackage{enumerate}
\usepackage{natbib}
\usepackage{url} 

\newcommand{\blind}{0}

\else
\documentclass[11pt]{article}
\usepackage[margin=1in]{geometry}
\usepackage[square]{natbib}
\usepackage{authblk}
\usepackage{xcolor}
\definecolor{red}{HTML}{F44336}
\definecolor{green}{HTML}{4CAF50}
\definecolor{yellow}{HTML}{FFEE58}
\definecolor{blue}{HTML}{0D47A1}
\usepackage[colorlinks,linkcolor=blue,citecolor=blue]{hyperref}
\fi

\usepackage{amsfonts}
\usepackage{amssymb}
\usepackage{amsmath}
\usepackage{amsthm}
\usepackage{graphicx,subcaption}
\usepackage{float} 
\usepackage{bm} 
\usepackage{bbm}
\usepackage{mathtools}
\usepackage{enumitem}
\usepackage{tablefootnote}
\usepackage{tikz}
\tikzstyle{mynode}=[thick,draw=blue,fill=blue!20,circle,minimum size=2]
\usepackage{xcolor}
\usepackage{longtable}
\usepackage[labelformat=simple]{subcaption}

\providecommand{\U}[1]{\protect\rule{.1in}{.1in}}
\theoremstyle{plain}
\newtheorem{theorem}{Theorem}

\newtheorem{lemma}[theorem]{Lemma}

\theoremstyle{definition}
\newtheorem{assumption}{Assumption}
\newtheorem{definition}{Definition}

\theoremstyle{remark}

\newtheorem{remark}{Remark}

\renewcommand{\P}{\mathbf{P}}
\newcommand{\E}{\mathbf{E}}

\newcommand{\0}{\bm{0}}

\newcommand{\bb}{\bm{b}}

\newcommand{\bu}{\bm{u}}
\newcommand{\bv}{\bm{v}}
\newcommand{\bw}{\bm{w}}
\newcommand{\bx}{\bm{x}}
\newcommand{\by}{\bm{y}}
\newcommand{\bz}{\bm{z}}

\newcommand{\bA}{\bm{A}}
\newcommand{\bI}{\bm{I}}
\newcommand{\bW}{\bm{W}}

\newcommand{\bLambda}{\bm{\Lambda}}

\newcommand{\Lb}{\overline{L}}

\newcommand{\II}{\mathbbm{I}}

\newcommand{\PP}{\mathbbm{P}}
\newcommand{\RR}{\mathbbm{R}}

\newcommand{\cF}{\mathcal{F}}
\newcommand{\cG}{\mathcal{G}}
\newcommand{\cH}{\mathcal{H}}

\newcommand{\cR}{\mathcal{R}}
\newcommand{\cS}{\mathcal{S}}

\newcommand{\cN}{\mathcal{N}}
\newcommand{\cX}{\mathcal{X}}
\newcommand{\cY}{\mathcal{Y}}

\newcommand{\cFt}{\widetilde{\mathcal{F}}}

\newcommand{\ttm}{\mathtt{m}}
\newcommand{\ttH}{\mathtt{H}}

\DeclareMathOperator{\sgn}{sgn}
\DeclareMathOperator{\diag}{diag}
\DeclareMathOperator{\VCdim}{VCdim}
\DeclareMathOperator{\Pdim}{Pdim}

\DeclareMathOperator*{\st}{s.t.}

\DeclarePairedDelimiter{\abs}{\lvert}{\rvert}
\DeclarePairedDelimiter{\norm}{\lVert}{\rVert}
\DeclarePairedDelimiter{\inner}{\langle}{\rangle}

\newcommand{\opt}{\mathtt{opt}}
\newcommand{\stat}{\mathtt{stat}}

\newcommand{\wt}{\widetilde}

\newcommand{\fh}{\hat{f}}
\newcommand{\ft}{\tilde{f}}

\newcommand{\gt}{\tilde{g}}

\newcommand{\wb}{\bar{w}}

\newcommand{\papertitle}{Representation-Enhanced Neural Knowledge Integration with Application to Large-Scale Medical Ontology Learning}

\begin{document}

\if\JASA1
\def\spacingset#1{\renewcommand{\baselinestretch}%
{#1}\small\normalsize} \spacingset{1}


\if1\blind
{
  \title{\bf \papertitle}
  \author{Suqi Liu\\
    Harvard University
    \bigskip\\
    Tianxi Cai\\
        Harvard University
            \bigskip\\
    Xiaoou Li \\
    University of Minnesota}
  \maketitle
} \fi

\if0\blind
{
  \bigskip
  \bigskip
  \bigskip
  \begin{center}
    {\LARGE\bf \papertitle}
\end{center}
  \medskip
} \fi

\bigskip
\else
\title{\papertitle}
\author[1]{Suqi Liu}
\author[1]{Tianxi Cai}
\author[2]{Xiaoou Li}
\affil[1]{Harvard University}
\affil[2]{University of Minnesota}
\date{}
\maketitle
\fi
\begin{abstract}
A large-scale knowledge graph enhances reproducibility in biomedical data discovery by providing a standardized, 
integrated framework that ensures consistent interpretation across diverse datasets.
It improves generalizability  by connecting data from various sources,
enabling broader applicability of findings across different populations and conditions.
Generating reliable knowledge graph, leveraging multi-source information from existing literature,
however, is challenging especially with a large number of node sizes and heterogeneous relations.  
In this paper, we propose a general theoretically guaranteed statistical framework, called
\textbf{R}epresentation-\textbf{E}nhanced \textbf{N}eural \textbf{K}nowledge \textbf{I}ntegration (RENKI),
to enable simultaneous learning of multiple relation types.
RENKI generalizes various network models widely used
in statistics and computer science.
The proposed framework incorporates representation learning output into initial entity embedding
of a neural network that approximates the score function for the knowledge graph
and continuously trains the model to fit observed facts.
We prove nonasymptotic
bounds for in-sample and out-of-sample weighted mean squared errors~(MSEs)
in relation to the pseudo-dimension of the knowledge graph function class.
Additionally,
we provide pseudo-dimensions for score functions based on multilayer neural networks
with rectified linear unit~(ReLU) activation function,
in the scenarios when the embedding parameters either fixed or trainable.
Finally, we complement our theoretical results with numerical studies
and apply the method to learn a comprehensive medical knowledge graph
combining a pretrained language model representation with knowledge graph links
observed in several medical ontologies.
The experiments justify our theoretical findings and demonstrate the effect of
weighting in the presence of heterogeneous relations and the benefit of 
incorporating representation learning in nonparametric models.
\end{abstract}

\noindent%
\if\JASA1
\vfill

\newpage
\spacingset{1.9} 
\fi

\section{Introduction}
Knowledge graph is a graph-structured model to represent human knowledge.
Entities such as objects, events, and concepts are symbolized as nodes,
and knowledge is stored as interlinked descriptions called relations
between these entities. A common data structure for knowledge graphs is a collection of factual triples
in the form of (\texttt{head}, \texttt{type}, \texttt{tail})
where both \texttt{head} and \texttt{tail} are entities in the knowledge graph
and \texttt{type} is one of the possible relations between entities.
Each triple forms a directed labeled edge for a pair of nodes
in the knowledge graph.
For example,
the triple (\textit{Obesity}, \textit{Causes}, \textit{Type 2 diabetes})
encodes the fact that the phenotype ``Obesity'' has a relation type ``Causes''
with the phenotype ``Type 2 diabetes''.
General purpose graph-based knowledge repositories such as DBpedia and Freebase 
emerged in early 2000 and were later commercialized by tech companies.
Most prominently, the Google Knowledge Graph builds on and largely expands
early knowledge bases by incorporating public resources
as well as licensed data.
Nowadays, knowledge graph sees broad applications in
information extraction~\citep{nickel2015review},
recommendation systems~\citep{zhang2016collaborative},
question answering~\citep{yao2014information},
and enhancing language models~\citep{chen2017neural}.

In the medical domain, structured knowledge bases store various relational facts among numerous clinical concepts
including disease phenotypes, signs and symptoms, drugs, procedures, and laboratory tests.
Key relations between them include ``is a'', ``associated with'', ``may treat'', ``may cause'',
and ``differential diagnosis''.
With the transition from traditional medical systems to modern electronic health records (EHRs)
since early 2000,
several domain-specific knowledge graphs have been developed,
such as the Systematized Nomenclature of Medicine Clinical Terms (SNOMED CT)~\citep{donnelly2006snomed}
for general clinical concepts,
Medication Reference Terminology (MED-RT), RxNorm~\citep{nelson2011normalized},
and Drug Side Effect Resource (SIDER)~\citep{kuhn2016sider} for medications,
and the Human Phenotype Ontology (HPO)~\citep{robinson2008human}
and PheCode~\citep{bastarache2021using} for disease phenotypes.
Over the past few decades, considerable manual efforts,
particularly by the National Library of Medicine (NLM),
have been devoted to assembling and integrating key terminologies and their relationships
into the Unified Medical Language System (UMLS)~\citep{bodenreider2004unified}. 
Despite these efforts, existing knowledge bases remain incomplete and noisy,
containing inaccurate or conflicting information in the curated relationships.
The vast number of clinical concepts, vocabulary heterogeneity,
and the complexity of medical relationships make manual curation of
large-scale knowledge bases challenging,
with issues of scalability and accuracy often compounded by human error.

With increasing information in the curated knowledge databases and advancement of machine learning
and statistical methods,
significant progress has also been made in completing biomedical knowledge graphs
by predicting the links between entities.
Most existing knowledge graph learning algorithms fall into a paradigm that the entities are embedded 
into a latent representation space and then the embedding vectors of the head and tail entities
are used as inputs of a relation-specific score function to jointly learn the embeddings
and the relation prediction function.
We refer the readers to the survey article~\citep{ji2021survey} for a review on some recent models.
Various score functions including neural networks have been proposed and applied
for many kinds of knowledge graphs.
In particular, Knowledge Vault~\citep{dong2014knowledge},
an early adoption of neural networks in the score function,
employs a multilayer perceptron~(MLP) with embedding for both entities and relations.
However, accurate link prediction for a large number of nodes remains highly challenging
especially in the presence of many relation types and significant sparsity of the observed links
relative to the total number of possible links.
As a result, model scalability is constrained by the limited number of observed triples (sample size). 
Since the embedding parameters grows linearly with the number of entities
in representation-based approaches,
successful model learning requires either a high number of observed relation pairs
between entities or a low embedding dimension.
This significantly limits the practical application of knowledge graph models
in real-world data.

One approach to overcome such sparse and noisy network data is to further leverage 
representation learning and knowledge fusion by integrating information
from additional sources.
In particular, large language models~(LLMs) such as BioBERT~\citep{lee2020biobert}, 
ClinicalBERT~\citep{huang2019clinicalbert}, and PubMedBERT~\citep{gu2021domain}
pre-trained on enormous biomedical text data through next token prediction~\citep{brown2020language}
supply good representations for many common biomedical entities.
These LLM-based represenations can serve as initialization of
the embedding parameters for the entities,
which can enhance the learning of the relations. 
To this end, we propose a general knowledge graph learning framework,
called \textbf{R}epresentation-\textbf{E}nhanced \textbf{N}eural \textbf{K}nowledge \textbf{I}ntegration (RENKI),
which combines statistical modeling with unsupervised representation learning.
The framework initializes entity embeddings from the outputs of representation learning algorithms like LLMs, 
offering flexibility in choosing score functions.
The model is trained using weighted least squares to account for the heterogeneity
in different types of relations.
We provide non-asymptotic bounds on both in-sample and out-of-sample weighted mean squared errors (MSEs),
in relation to the pseudo-dimension of the knowledge graph score function.
Additionally, we demonstrate the pseudo-dimension of multilayer ReLU networks with an embedding layer
for approximating the score function, instead of explicitly specifying it.
This offers a comprehensive theoretical understanding of the framework, enabling nonparametric model fitting. 

We further validate the theoretical findings through simulation studies and a real-world application.
The effectiveness of our two key components---sample weighting
and representation initialization---is demonstrated in simulations on synthetic data.
Additionally, we applied the RENKI algorithm to learn a large-scale medical knowledge graph
containing over $118,000$ clinical concepts,
encompassing both narrative and codified concepts from EHR data, across nine general relationship types.
Our algorithm successfully recovered observed facts in all relation types with high accuracy,
significantly outperforming existing methods.
The success of this real-world application highlights RENKI’s robust statistical guarantees
and its potential for completing large-scale biomedical knowledge graphs,
with broad implications for various downstream applications.

\subsection{Related work}
Knowledge graph representation learning has received wide attention from
both academia and industry in recent years.
Most of the methods focus on designing a score function based on entity (and
relation) embedding~\citep{bordes2013translating,ji2021survey}.
However, these models are largely inspired by empirical observations and lack theoretical guarantees.
In parallel, latent space models for networks have also attracted a long line of research in statistics.
Since the seminal work of \citet{hoff2002latent},
several variations of the model have been proposed and
analyzed~\citep{tang2013universally,ma2020universal}.
The latent space models cross with knowledge graphs
when they are extended to multilayer networks
(also called ``multiplex network'')~\citep{paul2020spectral,macdonald2022latent}.
Theoretical results,
including error bounds for model estimation (see, e.g., \citep{macdonald2022latent})
and statistical inference \citep{li2023statistical}, have been established.
One caveat of these works is that usually the connectivity between every pair of vertices
in every layer of the graph has to be observed.
This largely limits the applicability of the results in practice
when the graph is very sparse and contains significant missingness.
Moreover, the latent space model and its multilayer extensions rely on accurately specified score functions,
which can be overly restrictive, particularly when dealing with many types of relations.

In contrast, the RENKI framework allows flexible,
nonparametric score functions and accommodates general missing patterns,
supported by theoretical guarantees.
Specifically, we cast the RENKI approach as a nonparametric regression model tailored for knowledge graphs, 
deriving probabilistic error bounds for model estimation.
While theoretical investigation of ReLU network has been studied (see, e.g., \cite{bartlett1998almost,bartlett2019nearly,schmidt2020nonparametric,fan2023factor})
these results do not directly apply to RENKI due to the presence of the embedding layer
and multiple relation types.
We provide theoretical framework for the proposed model structure
better suited for large networks with sparse structure
and incorporating initial node embeddings from prior knowledge.
Beyond the theoretical analysis of RENKI,
some of our intermediate results extend traditional nonparametric regression theory
(see, e.g.,~\cite{gyorfi2002distribution}) by incorporating graph structures and heterogeneous relations, 
which may provide insights for related problems.

\subsection{Organization}
The remainder of the paper is organized as follows.
We first introduce notations used throughout the paper.
The methodologies are detailed in Section~\ref{sec:method} in three parts.
In Section~\ref{sec:model}, 
the statistical knowledge graph models are formally defined,
in particular a new function class of embedding-based neural networks.
We describe the parameter estimation method using weighted least squares in Section~\ref{sec:est}
and discuss combining representation learning in statistical knowledge graph models in Section~\ref{sec:init}. 
Finite sample theoretical results with oracle inequalities on the in-sample and out-of-sample error rates
are provided in Section~\ref{sec:theory}. 
Further, experimental results on synthetic data are presented in Section~\ref{sec:exp}.
An application of the proposed method to a real-world medical knowledge graph learning problem
is demonstrated in Section~\ref{sec:med}.
We conclude the paper by discussing the results and future directions
in Section~\ref{sec:disc}.

\subsection{Notations}
For a positive integer $N$, $[N] \coloneq \{1, \ldots, N\}$ is the set of all positive integers
up to $N$.
A $d$-dimensional vector or a tuple is denoted by a boldface character, e.g.,
$\bx = (x_1, \ldots, x_d)^{\top}$ where each $x_j$ is the $j$th element of $\bx$.
We treat all vectors as columns.
For two vectors $\bx, \by \in \RR^d$,
$\inner{\bx, \by} \coloneq \sum_{j=1}^d x_j y_j$ and 
$\bx \circ \by \coloneq (x_1 y_1, \ldots, x_d y_d)^{\top}$.
The Euclidean norm of a vector $\bx$ is defined as
$\norm{\bx} \coloneq (\sum_{j=1}^{d} x_j^2)^{1/2}$.
For a sequence of elements $x_1, \ldots, x_n$ in an arbitrary set $\cX$
and a function $f: \cX \to \RR$, 
denote $\norm{f}_{n} \coloneq (\frac{1}{n}\sum_{i=1}^{n} f(x_i)^2)^{1/2}$
the root mean square of $f(x_1), \ldots, f(x_n)$
and $\norm{f}_{\bw, n} \coloneq
(\frac{1}{n}\sum_{i=1}^{n} w(x_i) f(x_i)^2)^{1/2}$
the weighted version for a weight function $w: \cX \to \RR_+$.
The $\ell_2$-norm of $f$ under the probability measure $\PP$ is defined as
$\norm{f}_\PP \coloneq (\E_{X \sim \PP}[f(X)^2])^{1/2}$
and the weighted $\ell_2$-norm correspondingly as
$\norm{f}_{w, \PP} \coloneq (\E_{X \sim \PP}[w(X)f(X)^2])^{1/2}$.
For a function class $\cF$ with domain $\cY = \{-1, +1\}$,
$\VCdim(\cF)$ stands for its VC-dimension (see, e.g., \citep[Chapter 3.3]{anthony1999neural}).
With a slight abuse of notation,
$\VCdim(\cF)$ for a real-valued function is interpreted as the VC-dimension of
the function class $\{\sgn(f): f \in \cF\}$ where $\sgn(\cdot)$ is the sign function.
$\Pdim(\cF)$ denotes the pseudo-dimension (see, e.g., \citep[Definition~11.2]{anthony1999neural}) of a function class $\cF$.
We use the ``big-O'' notation,
where for two functions $f(n)$ and $g(n)$,
we write $f(n) = O(g(n))$ or $f(n) \lesssim g(n)$ if $|f(n)| \le C g(n)$
for a constant $C$ that does not depend on $n$.
We also write $f(n) = o(g(n))$ if $\lim_{n \to \infty} \frac{f(n)}{g(n)} = 0$.

\section{Methodologies}
\label{sec:method}
\subsection{Embedding-based statistical models for knowledge graphs}
\label{sec:model}
A common data structure to represent a knowledge graph $G$
on $N$ nodes with $n$ observed edges
is by a set of tuples $\{x_i = (h_i, r_i, t_i)\}_{i=1}^n$,
where $h_i, t_i \in [N]$
index the head and tail entities and $r_i \in [K]$ is the type of relation
among all $K$ possible choices.
Let $\cX \coloneq \{(h, r, t): h, t \in [N], r \in [K]\}$
be the set of all possible triples on $N$ entities with $K$ relations.
A knowledge graph score function
$\gamma: \cX \to \RR$ is a function that
takes a tuple as input and outputs a score
that determines the likelihood of the tuple.
Assume that we observe a sample set $\{(x_i, y_i)\}_{i=1}^n$ following
\begin{equation} \label{eq:kg}
y_i = \gamma(x_i) + \varepsilon_i,
\end{equation}
where $y_i$ represents the observed information about the tuple
such as a binary value indicating the relatedness between $h_i$ and $t_i$.
The noise $\varepsilon_i$'s are conditionally independent given $x_i$
and $\E[\varepsilon_i \mid x_i] = 0$.
Note that $\varepsilon_i$ enables us to easily incorporate many types of error,
such as the potential labeling error in the true relation $r_i$.

We further restrict the statistical knowledge graph models to these
based on embeddings and consider the family of score functions
\begin{equation} \label{eq:kge}
f(x = (h, r, t)) = f_r(\bz_h, \bz_t)  
\end{equation}
where $\bz_h, \bz_t \in \RR^{D}$
are the embedding vectors of the head and tail entities respectively,
and $f_k$ is a function defined for each relation type $k \in [K]$.
In general, both the node embeddings $\bz_1, \ldots, \bz_N$
and the relation functions $f_k(\cdot, \cdot)$'s need to be learned
although depending on what prior knowledge is available for $\bz_i$'s.
Different strategies, including the choice of the function class for $f_k$'s,
can be taken to train both $\bz_i$'s and $f_k$'s. 

The embedding-based knowledge graph model \eqref{eq:kge} covers
a wide range of methods previously proposed in the literature,
as evidenced by a recent survey on such models~\citep{ji2021survey}.
Table~\ref{tab:score} shows a few commonly adopted algorithms
with their corresponding score functions. 
\begin{table}[H]
\centering
\caption{Examples of knowledge graph models and their score functions.}
\label{tab:score}
\begin{tabular}[t]{c|c}
\hline
method & score function $f_r(\bz_h, \bz_t)$\\
\hline
inner product model\tablefootnote{In the original paper,
the inner product is normalized by the norm of tail embedding.
Here we use a simplified version studied in \cite{ma2020universal}.
Also, we ignore the degree heteorogeity parameters
which can be encoded as one coordinate of the embedding
and the covariates.}~\citep{hoff2002latent} & $\sigma(\bz_h^\top \bz_t)$\\
multilayer latent space model~\citep{zhang2020flexible} & $\sigma(\bz_h^\top \bLambda_r \bz_t)$\\
TransE~\citep{bordes2013translating} &  $-\norm{\bz_h - \bz_t - \bv_r}^2$\\
neural network model (MLP)~\citep{dong2014knowledge} & $\sigma(\bw^\top \sigma(\bW(\bz_h^\top, \bv_r^\top, \bz_t^\top)^\top))$\\
\hline
\end{tabular}
\end{table}

Instead of specifying a score function explicitly,
we fit it by assuming
\begin{equation*}
f(x = (h, r, t)) = g_r(\bz_h, \bz_t)
\end{equation*}
where $g_r$ is modeled as a feed-forward neural network.
For the neural network architecture, we focus on two models:
the Inner Product Neural Knowledge Graph (IP-NKG) and the Concatenation Neural Knowledge Graph (C-NKG).
Our real data analysis demonstrates that these models perform well in the application of harmonizing 
medical knowledge bases from diverse data sources.
Of note, the best neural network architecture is often application-dependent.
Our methods and theoretical results can be extended to other neural network architectures.

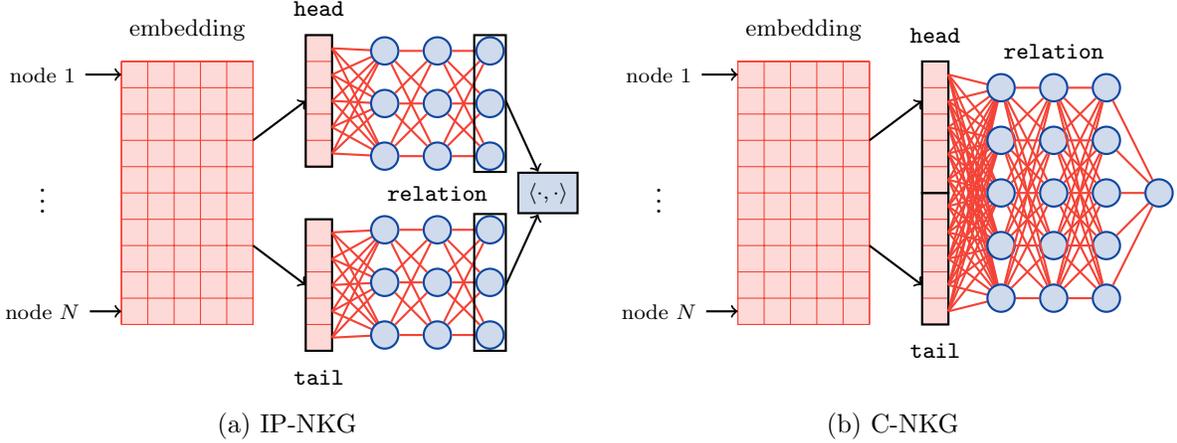
\begin{figure}
\centering
\begin{subfigure}[c]{0.48\textwidth}
\centering
\begin{tikzpicture}[scale=0.7]
\draw[step=0.5,red,thin,fill=red!20] (-4,-2.5) grid (-1.5,2.5) rectangle (-4,-2.5);
\node (n1) at (-5.5, 2.25) {\scriptsize node $1$};
\draw[->, thick] (n1) -- (-4,2.25);
\node at (-5.5,0) {$\vdots$};
\node (nn) at (-5.5, -2.25) {\scriptsize node $N$};
\draw[->, thick] (nn) -- (-4,-2.25);
\draw[step=0.5,red,thin,fill=red!20] (-0.5,0.5) grid (0,3) rectangle (-0.5,0.5);
\draw[step=0.5,red,thin,fill=red!20] (-0.5,-0.5) grid (0,-3) rectangle (-0.5,-0.5);
\draw[thick] (-0.5,0.5) rectangle (0,3);
\draw[thick] (-0.5,-0.5) rectangle (0,-3);
\draw[thick,->] (-1.5,1) -- (-0.5,1.75);
\draw[thick,->] (-1.5,-1) -- (-0.5,-1.75);
\foreach \i in {1,...,3}{
\foreach \j in {1,...,5}{
\draw[thick,red] (0,0.25+0.5*\j) -- (1,3.7-\i);
\draw[thick,red] (0,-0.25-0.5*\j) -- (1,0.3-\i);}}
\foreach \N [count=\lay,remember={\N as \Nprev (initially 0);}]
in {3,3,3}{ 
\foreach \i [evaluate={\y=\N/2-\i+2.2; \x=\lay; \prev=int(\lay-1);}]
  in {1,...,\N}{ 
\node[mynode] (N\lay-\i) at (\x,\y) {};
\ifnum\Nprev>0 
\foreach \j in {1,...,\Nprev}{ 
\draw[thick,red] (N\prev-\j) -- (N\lay-\i);}\fi}}
\foreach \N [count=\lay,remember={\N as \Nprev (initially 0);}]
in {3,3,3}{ 
\foreach \i [evaluate={\y=\N/2-\i-1.2; \x=\lay; \prev=int(\lay-1);}]
  in {1,...,\N}{ 
\node[mynode] (N\lay-\i) at (\x,\y) {};
\ifnum\Nprev>0 
\foreach \j in {1,...,\Nprev}{ 
\draw[thick,red] (N\prev-\j) -- (N\lay-\i);}\fi}}
\draw[thick] (2.7,0.4) rectangle (3.3,3);
\draw[thick] (2.7,-0.4) rectangle (3.3,-3);
\node[rectangle,draw,thick,fill=blue!20] (inn) at (4.1,0) {\scriptsize$\inner{\cdot, \cdot}$};
\draw[thick,->] (3.3,1.75) -- (inn);
\draw[thick,->] (3.3,-1.75) -- (inn);
\node[align=center] at (-2.75,3.1) {\footnotesize embedding};
\node at (-0.25,3.5) {\footnotesize\texttt{head}};
\node at (-0.25,-3.5) {\footnotesize\texttt{tail}};
\node at (2,0) {\footnotesize\texttt{relation}};
\end{tikzpicture}
\caption{IP-NKG}
\end{subfigure}
\begin{subfigure}[c]{0.48\textwidth}
\centering
\begin{tikzpicture}[scale=0.7]
\draw[step=0.5,red,thin,fill=red!20] (-4,-2.5) grid (-1.5,2.5) rectangle (-4,-2.5);
\node (n1) at (-5.5, 2.25) {\scriptsize node $1$};
\draw[->, thick] (n1) -- (-4,2.25);
\node at (-5.5,0) {$\vdots$};
\node (nn) at (-5.5, -2.25) {\scriptsize node $N$};
\draw[->, thick] (nn) -- (-4,-2.25);
\draw[step=0.5,red,thin,fill=red!20] (-0.5,0) grid (0,2.5) rectangle (-0.5,0);
\draw[step=0.5,red,thin,fill=red!20] (-0.5,0) grid (0,-2.5) rectangle (-0.5,-0);
\draw[thick] (-0.5,0) rectangle (0,2.5);
\draw[thick] (-0.5,0) rectangle (0,-2.5);
\draw[thick,->] (-1.5,1) -- (-0.5,1.75);
\draw[thick,->] (-1.5,-1) -- (-0.5,-1.75);
\foreach \i in {1,...,5}{
\foreach \j in {1,...,5}{
\draw[thick,red] (0,-0.25+0.5*\j) -- (1,3-\i);
\draw[thick,red] (0,0.25-0.5*\j) -- (1,3-\i);}}
\foreach \N [count=\lay,remember={\N as \Nprev (initially 0);}]
in {5,5,5,1}{ 
\foreach \i [evaluate={\y=\N/2-\i+0.5; \x=\lay; \prev=int(\lay-1);}]
  in {1,...,\N}{ 
\node[mynode] (N\lay-\i) at (\x,\y) {};
\ifnum\Nprev>0 
\foreach \j in {1,...,\Nprev}{ 
\draw[thick,red] (N\prev-\j) -- (N\lay-\i);}\fi}}
\node[align=center] at (-2.75,3.1) {\footnotesize embedding};
\node at (-0.25,3) {\footnotesize\texttt{head}};
\node at (-0.25,-3) {\footnotesize\texttt{tail}};
\node at (-0.25,3.5) {\footnotesize\phantom{\texttt{head}}};
\node at (-0.25,-3.5) {\footnotesize\phantom{\texttt{tail}}};
\node at (2,2.7) {\footnotesize\texttt{relation}};
\end{tikzpicture}
\caption{C-NKG}
\end{subfigure}
\caption{Schematic diagrams of neural knowledge graph models.
Red blocks represent trainable embedding parameters.
Red lines represent trainable weights.
Blue circles and blocks stand for values after the operations.}
\label{fig:nkg}
\end{figure}

Schematic diagrams illustrating IP-NKG and C-NKG models can be found in Figure~\ref{fig:nkg}.
Their precise definitions are provided below.
In IP-NKG, the head and tail embeddings are passed through two separate
feed-forward networks and then the inner product of the outputs are taken.
In C-NKG, the head and tail embeddings are first concatenated
into one vector and then passed through the feed-forward networks.
First, let us introduce the feed-forward neural network, a key component of both models.
\begin{definition}[Feed-forward network]
\label{def:ffn}
A feed-forward network is defined as follows.
\begin{itemize}[nosep]
\item The input is a $H^{(0)}$-dimensional real vector.
\item A repetition of $L-1$ feed-forward layers are stacked sequentially.
Let $H^{(1)}, \ldots, H^{(L-1)}$ be a sequence of positive integers.
For $\ell = 1, \ldots, L-1$, the $\ell$th feed-forward layer is a function
$g^{(\ell)} : \RR^{H^{(\ell-1)}} \to \RR^{H^{(\ell)}}$ such that
\begin{equation*}
g^{(\ell)}(\bx) = \eta(\bA^{(\ell)} \bx + \bb^{(\ell)})
\end{equation*}
where $\bA^{(\ell)} \in \RR^{H^{(\ell-1)} \times H^{(\ell)}}$ is the weight matrix,
$\bb^{(\ell)} \in \RR^{H^{(\ell)}}$ is the bias vector,
and $\eta: \RR \to \RR$ is the activation function.
\item The output layer, denoted by $g_L$, is a linear function
\begin{equation*}
g^{(L)}(\bx) = \bA^{(L)} \bx + \bb^{(L)}
\end{equation*}
where $\bA^{(L)} \in \RR^{H^{(L-1)} \times H^{(L)}}$ is the weight matrix
and $\bb^{(L)}$ is the bias vector.
\item A feed-forward network $g: \RR^{H^{(0)}} \to \RR^{H^{(L)}}$ is a composition function
defined as
\begin{equation*}
g = g^{(L)} \circ \cdots \circ g^{(1)}.
\end{equation*}
\end{itemize}
\end{definition}
As the number of layers or the width of a feed-forward neural network increases,
it can approximate a wide class of continuous functions.
Specifically, see, e.g., \cite{yarotsky2017error} for approximation theory for ReLU networks.
Note that we allow some of the weight parameters to be fixed constants,
in particular $0$,
which results in partially connected networks.

\begin{definition}[Inner Product Neural Knowledge Graph Model (IP-NKG)]
\label{def:ipkg}
The inner product neural knowledge graph model is defined as follows.
\begin{itemize}[nosep]
\item Each entity $i = 1, \ldots, N$ is associated with a $D$-dimensional vector
$\bz_i \in \RR^D$.
\item For each relation type $k \in [K]$, there are two feed-forward networks $g_k$ and $g_k'$,
both with input dimension $D$ and output dimension $D'$.
\item The score function is then defined as
\begin{equation*}
f(x=(h, r, t)) = \rho(g_r(\bz_h)^\top g_r'(\bz_t))
\end{equation*}
where $\rho$ is a monotone function.
\end{itemize}
\end{definition}

\begin{definition}[Concatenation Neural Knowledge Graph Model (C-NKG)]
\label{def:cnkg}
The concatenation neural knowledge graph model is defined as follows.
\begin{itemize}[nosep]
\item Each entity $i = 1, \ldots, N$ is associated with a $D$-dimensional vector
$\bz_i \in \RR^D$.
\item For each triple $(h, r, t)$, we concatenate the head and tail embeddings into 
$\bx = (\bz_h^{\top},\bz_t^{\top})^{\top}$.
\item The concatenated vector $\bx$ is then passed through a relation specific function
$g_k(\bx) = g_k(\bz_h, \bz_t)$ for $k \in [K]$ where $g_k(\cdot)$ is a feed-forward network
with input dimension $2D$ and output dimension $1$.
\item The score function is then defined as
\begin{equation*}
f(x=(h, r, t)) = \rho(g_r(\bz_h, \bz_t))
\end{equation*}
where $\rho$ is a monotone function.
\end{itemize}
\end{definition}

The two neural knowledge graph models in Definition~\ref{def:ipkg} and~\ref{def:cnkg}
generalize many score functions proposed in the literature,
in particular these listed in Table~\ref{tab:score}.
With specifically chosen weight matrices,
IP-NKG includes
the inner product model and multilayer latent space model,
and C-NKG includes TransE and the MLP model.

For model training, activation functions $\eta$ and $\rho$ are pre-specified.
Generally, the weights $\bA^{(\ell)}$'s, the biases $\bb^{(\ell)}$'s,
and the embeddings $\bz_i$'s are unknown and learned from the data.
Some of them can also be fixed to constants thus not contributing to the model parameters.
For simplicity of presentation and practical usage,
we focus on neural network models with rectified linear unit (ReLU) activation function  
$\eta(x) = \max(0,x)$ 
for the rest of the paper unless otherwise specified,
although our proofs apply to piecewise polynomial activation functions in general.
Each model in Definition~\ref{def:ipkg} and~\ref{def:cnkg} defines
a corresponding function class.
We refer to both of them as the neural knowledge graph function class.

\subsection{Parameter estimation via weighted least squares}
\label{sec:est}
Suppose we are given a set of samples $\{(x_i, y_i)\}_{i=1}^n$
where $y_i$ follows the definition \eqref{eq:kg}.
Our goal is to approximate the knowledge graph score function $\gamma$
as close as possible under some appropriate error measurement.
To this end, we define the weighted empirical risk (also referred to as loss)
\begin{equation} \label{eq:emp_risk}
\cR_{\bw, n}(f)
\coloneq \frac{1}{n}\sum_{i=1}^{n} w(x_i) (y_i - f(x_i))^2
\end{equation}
where $w: \cX \to \RR^+$ is a weight function depending only on $x_i$'s.
The weights compensate the heteroskedasticity in the sample
and hence are chosen such that the instances with large variance
are weighted down and these with small variance are weighted up.
When discussing the in-sample MSE,
without loss of generality, we may assume that
$\wb \coloneq \frac{1}{n} \sum_{i=1}^{n} w(x_i) = 1$
since the only difference will be a scaling constant in front of
the weighted MSE.
For simplicity of notation,
we use $w_i$ and $w(x_i)$ interchangeably.

We assume that there exists an optimization algorithm that obtains
the approximate empirical risk minimizer $\hat{f} \in \cH$ such that
\begin{equation} \label{eq:approx_sol}
\cR_{\bw, n}(\hat{f}) \le \inf_{f \in \cH} \cR_{\bw, n}(f) + \delta_\opt
\end{equation}
for some optimization error $\delta_\opt \ge 0$.
For example, one can use gradient-based algorithms
with a carefully chosen starting point.

\subsection{Initialization with representation learning}
\label{sec:init}
As we shall see from both theoretical and empirical results
in the following sections,
directly minimizing the weighted MSE would produce poor generalization
when the sample size is small
(usually on the same order or smaller than the number of parameters).
Therefore, we propose to initialize the embedding parameters
by representation learning from other sources of data.

To be more specific,
suppose that each entity in the knowledge graph
are not only identified by its connection to other entities
in the knowledge graph
but also associated with some side information such as its semantic meaning.
For instance, in a medical knowledge graph,
each concept typically has one or more text descriptions associated with it.
These descriptions provide valuable information that either aligns with
or enhances the knowledge captured by the graph.
A powerful approach to harness this information is by representing the entities
using large language models pre-trained on biomedical text,
which can embed the semantic meaning of the text descriptions into a rich,
informative representation of the entities.

Let $e: [N] \to \RR^D$ be some representation learning algorithm
that maps each entity in the knowledge graph to a $D$-dimensional vector using its text description.
We then initialize the parameters of the embedding layer of the neural knowledge graph models
with $\bz_i^0 = e(i) \quad (i = 1, \ldots, N)$
from the outputs of the representation learning algorithm for each entity.
Then the optimization algorithm such as gradient descent proceeds
as usual until some stopping criteria.
In particular,
the employment of neural architecture in the model allows for the use of modern large scale
parallel computation platforms such as PyTorch,
which computes the gradient distributedly on GPUs that greatly speeds up the training process.

\section{Theoretical results}
\label{sec:theory}
We first state the assumptions on the knowledge graph data.
\begin{assumption} [Sub-Gaussian noise] \label{assump:subGaussian}
Conditioned on $x_1, \ldots, x_n$,
the noise $\varepsilon_1, \ldots, \varepsilon_n$ are independent sub-Gaussian
with variance proxies $\sigma_1^2, \ldots, \sigma_n^2$ respectively.
That is, for $i = 1, \ldots, n$,
there exists a constant $\sigma_i$ such that
$\E[e^{\lambda \varepsilon_i} \mid x_i] \le e^{\lambda^2\sigma_i^2/2}$
for all $\lambda \in \RR$ almost surely.
\end{assumption}
We also require the boundedness of the underlying knowledge graph score function.
\begin{assumption} [Boundedness] \label{assump:bound}
The score function of the knowledge graph
$\gamma: \cX \to \RR$ is uniformly bounded in $\cX$,
i.e., $\sup_{x \in \cX} \abs{\gamma} \le B$.
\end{assumption}
The previous two assumptions are sufficient to warrant a bound on 
the in-sample weighted MSE.
However, in order to derive out-of-sample MSE bounds,
we additionally require the samples to be independent
and identically distributed~(i.i.d.) from some (possibly unknown) distribution.
\begin{assumption} [Independent and identically distributed samples]
\label{assump:iidsample}
The triples $x_1, \ldots, x_n$ are i.i.d.\ samples following the distribution $\PP$
on all possible tuples $\cX \coloneq \{(i, j, k): i, j \in [N], k \in [R]\}$.
In other words,
$x_1, \ldots, x_n \stackrel{i.i.d.}{\sim} \PP(\cX)$.
\end{assumption}
\begin{remark}
Since the in-sample MSE only depends on the observed data
$\{(x_i, y_i)\}_{i=1}^n$,
specific distribution of $x_1, \ldots, x_n \in \cX$ is not required.
However, when studying generalization to new samples,
we need the connection between training and testing distributions.
For simplicity,
we assume that the $x_i$'s are i.i.d.\ 
and it is only needed for the out-of-sample bound of Theorem~\ref{thm:os_mse}.
\end{remark}

We establish high-probability bounds on the following in-sample
weighted mean squared error~(MSE)
\begin{equation*}
\norm{\fh - \gamma}_{\bw, n}^2
\coloneq \frac{1}{n}\sum_{i=1}^{n} w_i \abs{\fh(x_i) - \gamma(x_i)}^2
\end{equation*}
and the out-of-sample weighted mean squared error~(MSE)
\begin{equation*}
\norm{\fh - \gamma}_{w, \PP}^2
\coloneq \E[w(x) \abs{\fh(x) - \gamma(x)}^2].
\end{equation*}

Before presenting our main results,
we define several useful quantities.
Let $\sigma_\ttm^2 \coloneq \max_{i \in [n]} w_i \sigma_i^2$ be
the largest effective variance in the data
and $w_\infty \coloneq \sup_{x \in \cX} w(x)$ be the upper bound of
the weight function.

We first present our general results providing the oracle inequalities
for both the in-sample and out-of-sample weighted MSEs
that are applicable to all function classes characterized
by their pseudo-dimensions.
\begin{theorem}[Oracle inequalities] \label{thm:oracle}
Let $p \coloneq \Pdim(\cH)$ denote the pseudo-dimension of
the function class $\cH$, and $\fh$ satisfies \eqref{eq:approx_sol}. 
\begin{enumerate}[label={(\roman*)},ref={\thetheorem(\roman*)},itemindent=1em]
\item\label{thm:is_mse}
Suppose that Assumptions~\ref{assump:subGaussian} and~\ref{assump:bound}
hold.
With probability at least $1 - 2(\frac{p}{en})^{p}$,
\begin{equation*}
\norm{\fh - \gamma}_{\bw, n}^2
\le 3 \inf_{f \in \cH} \norm{f - \gamma}_{\bw, n}^2
+ \frac{(36(2+\wb) \sigma_\ttm^2 + 16B^2)p}{n} \log \biggl(\frac{en}{p}\biggr)
+ 2\delta_\opt.
\end{equation*}
Let $\sigma_\ttH^2 \coloneq
\bigl(\frac{1}{n}\sum_{i=1}^{n} \frac{1}{\sigma_i^2}\bigr)^{-1}$
be the harmonic mean of $\sigma_1^2, \ldots, \sigma_n^2$.
By choosing $w(x_i) = \sigma_\ttH^2/\sigma_i^2$,
we have that
\begin{equation*}
\norm{\fh - \gamma}_{\bw, n}^2
\le 3 \inf_{f \in \cH} \norm{f - \gamma}_{\bw, n}^2
+ \frac{4(27 \sigma_\ttH^2 + B^2) p}{n} \log \biggl(\frac{en}{p}\biggr)
+ 2\delta_\opt.
\end{equation*}
\item \label{thm:os_mse}
Additionally suppose that Assumption~\ref{assump:iidsample} holds.
Then with probability at least $1 - 10(\frac{p}{en})^{p}$,
\begin{equation*}
\norm{\fh - \gamma}_{w, \PP}^2
\le 27 \inf_{f \in \cH} \norm{f - \gamma}_{w, \PP}^2
+ \frac{(108(2 + \wb) \sigma_{\ttm}^2 + 21760 w_\infty B^2 + 48 B^2)p}{n}
\log \biggl(\frac{en}{p}\biggr) + 6\delta_\opt.
\end{equation*}
By taking the weight function $w(x) = B^2/\max\{\sigma(x)^2, B^2\}$,
we have that
\begin{equation*}
\norm{\fh - \gamma}_{w, \PP}^2
\le 27 \inf_{f \in \cH} \norm{f - \gamma}_{w, \PP}^2
+ \frac{22132 B^2 p}{n} \log \biggl(\frac{en}{p}\biggr)
+ 6\delta_\opt.
\end{equation*}
\end{enumerate}
\end{theorem}

\begin{remark}
We see that the choices of weights for in-sample and out-of-sample MSEs are
slightly different.
For the in-sample MSE,
the weights are chosen as $w_i \propto 1/\sigma_i^2$ to precisely
match for the heteroskedasticity in the sample.
In doing so,
the largest effective variance
$\sigma_\ttm^2 \coloneq \max_{i \in n} w_i \sigma_i^2$ is minimized.
For the out-of-sample MSE,
the weights are chosen depending on not only the sample variance,
but also the value range of the score function.
This choice is more conservative for samples with small variances.
There is an interesting tradeoff between the largest effective variance $\sigma_\ttm^2$
and the maximum weight $w_\infty$.
The weights are optimized such that the combination of the two is the minimized.
\end{remark}

Together with pseudo-dimensions for particular function classes of knowledge graph models,
oracle inequalities are readily available by applying Theorem~\ref{thm:oracle}.
We show the pseudo-dimensions for the models based on feed-forward networks with ReLU activation
in two separate scenarios: the embedding being fixed and trainable.
Note that many other score functions in the literature
either are special cases of the two neural knowledge graph models
or have architectures that can be studied similarly.
We remark on this after the lemmas.

We first study the fixed embedding case with the help of the following lemma
which upper bounds the mixture of experts~(MOE) type of models
with a designated expert for each sample.
\begin{lemma}[Mixture of experts]
\label{lem:composition}
Let $\cH_1, \ldots, \cH_K$ be $K$ function classes
where each $f_k \in \cH_k$ maps $\cX$ to $\{0, 1\}$.
We consider a function class $\cH$ which is a composition of $\cH_k$
such that for $f \in \cH$,
$f((x, k)) = f_k(x)$ maps $\cX \times [K]$ to $\{0, 1\}$.
Then the VC-dimension of $\cH$ satisfies
\begin{equation*}
\VCdim(\cH) \le 4 \sum_{k=1}^K \VCdim(\cH_k).
\end{equation*}
\end{lemma}

When the embeddings are fixed,
this lemma, combined with an upper bound on the VC-dimension of particular function classes,
would imply the VC-dimension of the score functions.
For example, the VC-dimension of piecewise linear neural networks~\citep{bartlett2019nearly}
can be directly applied to the concatenation model
(and the inner product model with a slight modification).
By the connection between VC-dimension and pseudo-dimension for neural networks
(see Lemma~\ref{lem:fixed} in the appendix and the remark therein),
we present the following lemma. 

\begin{lemma}[Pseudo-dimension with fixed embedding] \label{lem:Pdim_fix}
Suppose that the knowledge graph has $N$ nodes and $K$ relations.
Consider the neural knowledge graph model
in Definition~\ref{def:cnkg} with fixed embeddings
and ReLU activation function.
That is, $\bz_h$ and $\bz_t$ are treated as fixed vectors,
and the weights $\bA^{(\ell)}$'s and biases $\bb^{(\ell)}$'s in the neural network are unknown
and to be learned.
Let $L_k$ be the number of layers and $W_k$ be the total number of weights
in the ReLU network for the $k$th relation.
Then the pseudo-dimension of the neural knowledge graph function class $\cH$ satisfies
\begin{equation*}
\Pdim(\cH) \lesssim \sum_{k=1}^K L_k W_k \log W_k.
\end{equation*}
\end{lemma}

For neural knowledge graph models with trainable embedding,
we derive the following upper bound on the pseudo-dimension.
\begin{lemma}[Pseudo-dimension with trainable embedding] \label{lem:Pdim}
Suppose that the knowledge graph has $N$ nodes and $K$ relations.
Consider the neural knowledge graph models in Definition~\ref{def:ipkg} and~\ref{def:cnkg}.
The pseudo-dimension of the neural knowledge graph function class $\cH$
with embedding dimension $D$, max number of layers $L \coloneq \max_k L_k$,
and average total number of feed-forward network parameters $W \coloneq \frac{1}{K} \sum_{k=1}^K W_k$
satisfies
\begin{equation*}
\Pdim(\cH) \lesssim (N D + K W) L \log (K W).
\end{equation*}
\end{lemma}
\begin{remark}
Pseudo-dimension bounds for some other embedding-based knowledge graph models
can be similarly obtained from the proof of Lemma~\ref{lem:Pdim}.
For example, the TransE model would have pseudo-dimension $O((N+K)D \log (KD))$.
The multilayer latent space model would have pseudo-dimension $O((N + KD)D \log (KD))$.
If instead of using a feed-forward network for each relation,
we create a vector of the same dimension of the node embeddings
and use the same neural network for all relations,
this generalizes the MLP model of Knowledge Vault~\citep{dong2014knowledge},
which employs only one hidden layer in the neural network.
Using the proof of Lemma~\ref{lem:Pdim},
this alternative neural network model would have pseudo-dimension
$O((N D+ K D + W) L \log W)$ where $W$ is the total number of parameters
in the single network.
We conclude this remark by noting that upper bounds on the pseudo-dimensions
of many more models can be similarly obtained by the results.
\end{remark}

Combining Theorem~\ref{thm:oracle} and Lemma~\ref{lem:Pdim},
we have the following theorem.
\begin{theorem}[Oracle inequality for the knowledge graph model]
Let $\fh$ satisfy \eqref{eq:approx_sol} with the function class
$\cH$ being the neural knowledge graph model in Definition~\ref{def:ipkg} or Definition~\ref{def:cnkg}.
\begin{enumerate}[label={(\roman*)},ref={\thetheorem(\roman*)},itemindent=1em]
\item\label{thm:is_mse_nn}
Suppose that Assumption~\ref{assump:subGaussian} and~\ref{assump:bound} hold
and let $\sigma_\ttH^2 \coloneq
\bigl(\frac{1}{n}\sum_{i=1}^{n} \frac{1}{\sigma_i^2}\bigr)^{-1}$
be the harmonic mean of $\sigma_1^2, \ldots, \sigma_n^2$.
Then, with probability $1 - o(1)$,
\begin{equation*}
\norm{\fh - \gamma}_{\bw, n}^2
\lesssim \inf_{f \in \cH} \norm{f - \gamma}_{\bw, n}^2 + \delta_\stat + \delta_\opt
\end{equation*}
where
\begin{equation*}
\delta_\stat =  \frac{(\sigma_\ttH^2 + B^2) (N D + K W) L
\log (K W)}{n} \log \biggl(\frac{n}{(N D + K W) L}\biggr).
\end{equation*}
\item \label{thm:os_mse_nn}
Additionally suppose that Assumption~\ref{assump:iidsample} holds
and let $w(x) = B^2 / \max\{\sigma(x)^2, B^2\}$.
Then, with probability $1 - o(1)$,
\begin{equation*}
\norm{\fh - \gamma}_{w, \PP}^2
\lesssim \inf_{f \in \cH} \norm{f - \gamma}_{w, \PP}^2 + \delta_\stat + \delta_\opt
\end{equation*}
where
\begin{equation*}
\delta_\stat = \frac{B^2 (N D + K W) L \log (K W)}{n}
\log \biggl(\frac{n}{(N D + K W) L}\biggr).
\end{equation*}
\end{enumerate}
\end{theorem}

We note that in the above theorem, $B,N,D,K,W$ and $L$ can grow as a function of the sample size $n$.
We also note a trade-off between model complexity and approximation error:
as $D$, $L$, and $W$ increase,
the approximation errors $\inf_{f \in \cH} \norm{f - \gamma}_{\bw, n}^2$
and $\inf_{f \in \cH} \norm{f - \gamma}_{w, \PP}^2$ of the ReLU network decrease
(see, e.g., \cite{yarotsky2017error}, for results on approximation error for deep ReLU networks).
However, this comes at the cost of a larger pseudo-dimension, which increases $\delta_\stat$. 
If the particular function class of $\gamma$ is known,
it is possible to choose the pseudo-dimension of the neural network to optimize the error rates.
Roughly, more complex function classes would require higher pseudo-dimensions of the neural network.

\section{Simulation studies}
\label{sec:exp}
Experiments on synthetically generated data are carried out
to justify the theoretical results
as well as to compare with existing methods.
We first create data following the assumptions
made by the theoretical results and consider regression problems
for two different function classes.
The main goals are to demonstrate the effect of weighting
when there is heterogeneous relations
and to show the benefits of initialization
in the overparametrized regime.
We then build synthetic knowledge graphs according to two nonlinear generative models
and reveal samples from only positive edges.
The proposed NKG model and the benchmark TransE algorithm are compared
on this binary knowledge graph data.
Throughout simulation studies, we train the model using the C-NKG architecture described
in Definition~\ref{def:cnkg}, where the activation function $\eta$ is the ReLU function,
the monotone function $\rho$ is the identity function $\rho(x) = x$,
and the feed-forward neural networks are fully connected.

\subsection{Effect of weighting}
We first demonstrate how weighting can affect the MSEs.
The data is generated using two different models:
a concatenated linear model and a nonlinear vector offset model.
The concatenated linear model is defined as
\begin{equation*}
\gamma(x = (h, r, t)) = (\bu_h^\top, \bu_t^\top) \bv_r
\end{equation*}
where $\bu_h, \bu_t \in \RR^d$ are the embedding vectors for the head and tail
nodes and $\bv_k \in \RR^{2d}$ is the weight parameters for relation type $k$.
Note that the concatenated linear model is in the C-NKG function class
in Definition~\ref{def:cnkg},
and hence the model is correctly specified.
We also consider a nonlinear vector offset model which is inspired by the linguistic regularity
in word represenations~\citep{mikolov2013linguistic}
and is similar to the loss of TransE~\citep{bordes2013translating}.
To be specific, the vector offset model is defined by
\begin{equation*}
\gamma(x = (h, r, t)) = -\norm{\bu_h - \bu_t + \bv_r}^2.
\end{equation*}
Note that the vector offset model \emph{cannot} be represented by C-NKG
with ReLU activation function,
and hence we have a misspecified model
(i.e., $\inf_{f \in \cH} \norm{f - \gamma}_{w, \PP} > 0$).
The observed samples $\{(x_i, y_i)\}_{i=1}^n$ follow the definition in \eqref{eq:kg},
i.e., $y_i = \gamma(x_i) + \varepsilon_i$.
We minimize the weighted or unweighted (i.e., $w(x_i) = 1$ for $i=1, \ldots, n$)
empirical risks as defined in~\eqref{eq:emp_risk} using gradient-based optimization.

For both the concatenated linear model and the vector offset model,
we fix the number of nodes $N=500$, the embedding dimension $d=20$,
and the number of relations $K=5$.
Each $\bu_i$ is an independent $d$-dimensional standard normal
vector, i.e., $\bu_i \sim \cN(\0, \bI_d)$.
For the concatenated linear model,
each $\bv_k$ is an independent $2d$-dimensional standard normal vector.
For the vector offset model,
the embeddings are generated similarly,
and each $\bv_k$ is an independent $d$-dimensional standard normal vector.
The noise $\varepsilon_i$ for each sample is an independent Gaussian random variable,
with a standard deviation $1$ or $5$,
chosen uniformly at random and independently of other variables. 

In both the concatenated linear and vector offset models,
the embedding layer of the neural networks is set
to have the same dimension $D = d$ as the data.
For the linear model,
we use a two-layer ReLU C-NKG with one hidden layer of $32$ units to fit the data.
The sample size $n$ changes from $20,000$ to $40,000$ in $5,000$ increments.
The test data is another $40,000$ independent random samples from the model with
the same parameters.
For the vector offset model, 
we use a ReLU C-NKG with two hidden layers of
$256$ and $100$ hidden units respectively.
The sample size $n$ changes from $20,000$ to $60,000$ in $10,000$ increments.
The test data is again another $100,000$ independent random samples
from the model with the same parameters.
Each experiment is performed for $10$ independent runs,
and we report their mean and standard deviation.

We train the model with both weighted and unweighted empirical risks
and evaluate its out-of-sample MSEs respectively.
The results are plotted in Figure~\ref{fig:mse_n}.
\begin{figure}[ht]
\centering
\begin{subfigure}{0.48\textwidth}
\includegraphics[width=\textwidth]{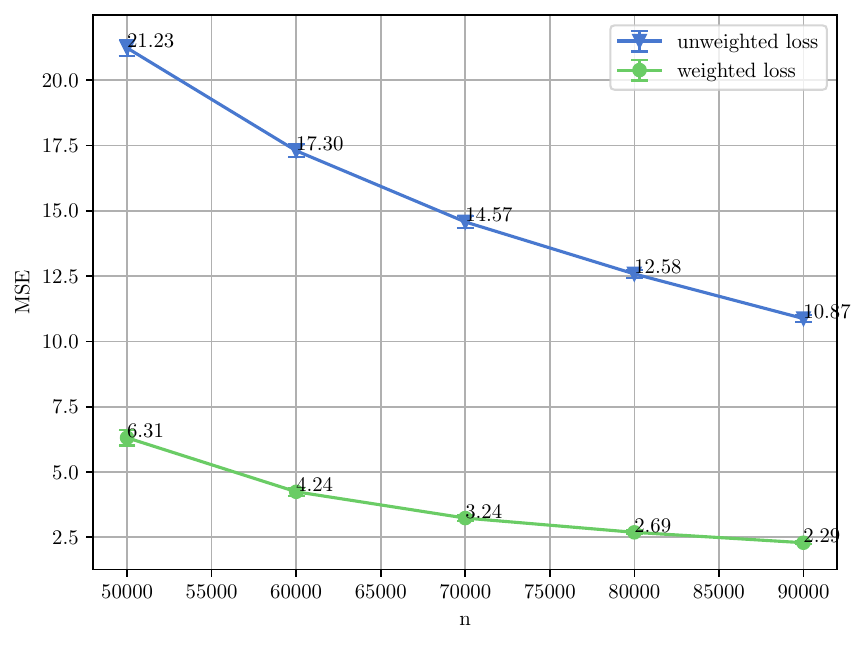}
\caption{Concatenated linear model.}
\end{subfigure}
\begin{subfigure}{0.48\textwidth}
\includegraphics[width=\textwidth]{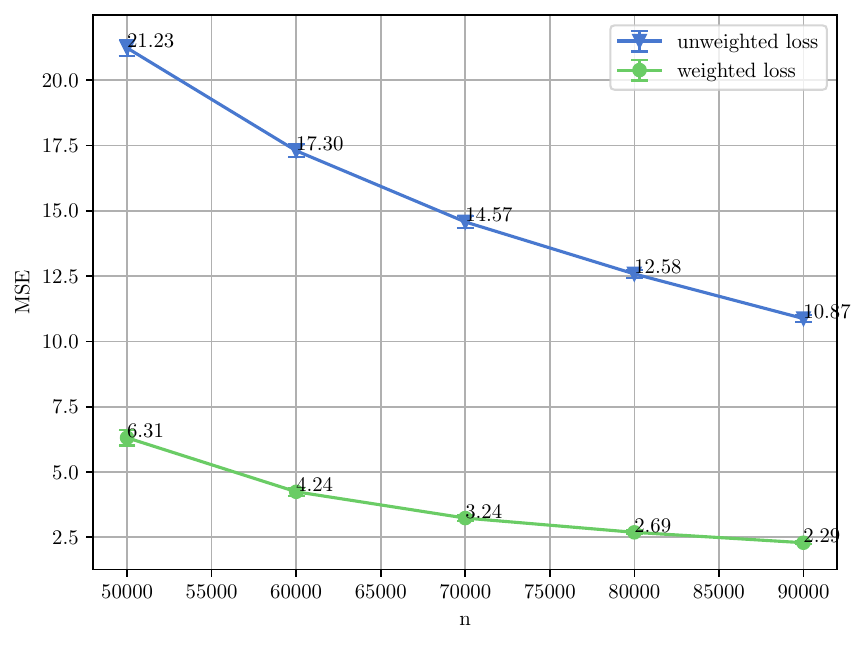}
\caption{Vector offset model.} 
\end{subfigure}
\caption{Effect of sample sizes.
We report the mean of $10$ independent random runs
and the error bars represent the standard deviation calculated from them.}
\label{fig:mse_n}
\end{figure}

It is clear that the out-of-sample MSEs are much smaller
when the model is trained by minimizing the weighted loss.
For both the concatenated linear and vector offset model,
the MSEs exhibit a $\frac{1}{n}$-rate of convergence in the sample size $n$,
while a larger sample size is needed for the vector offset model
due to the neural network's increased number of parameters.

\subsection{Effect of initialization}
Next we investigate how the initialization of the embedding layer would affect
the MSEs.
In this set of experiments,
the concatenated linear model is used to generate the knowledge graphs.
We again fix the number of nodes $N=500$,
the embedding dimension $d=20$, and the number of relations $K=5$.
Two distinct regimes are studied here:
the parametric regime where the sample size is much larger
than the number of parameters;
the overparametrized regime where the sample size is much smaller
compared to the number of parameters.
Note that the total number of parameters is $500 \times 20 \text{ (embeddings)} + 21 \times 32 \text{ (first-layer weights and biases)} + 33 \text{ (second-layer weights and bias)} = 10,705$.
We experiment with three different strategies:
1.~initialize the embedding randomly (random);
2.~initialize the embedding with the true parameters that generated the data
but continue to train these parameters (init);
3.~initialize the embedding with the true parameters but keep them fixed
during training (freeze).
The results are shown in Figure~\ref{fig:mse_large}.
We additionally run experiment on much smaller sample sizes.
For the small sample size experiments,
in addition to initializing embeddings with standard Gaussian vectors (random)
and with truth (init),
we perturb the true embeddings with Gaussian noise and use them as initialization
(noise).
The Gaussian noise has variance $1$, hence the signal-to-noise ratio is $1$.
The results are plotted in Figure~\ref{fig:mse_small}.

\begin{figure}[ht]
\centering
\begin{subfigure}{0.48\textwidth}
\includegraphics[width=\textwidth]{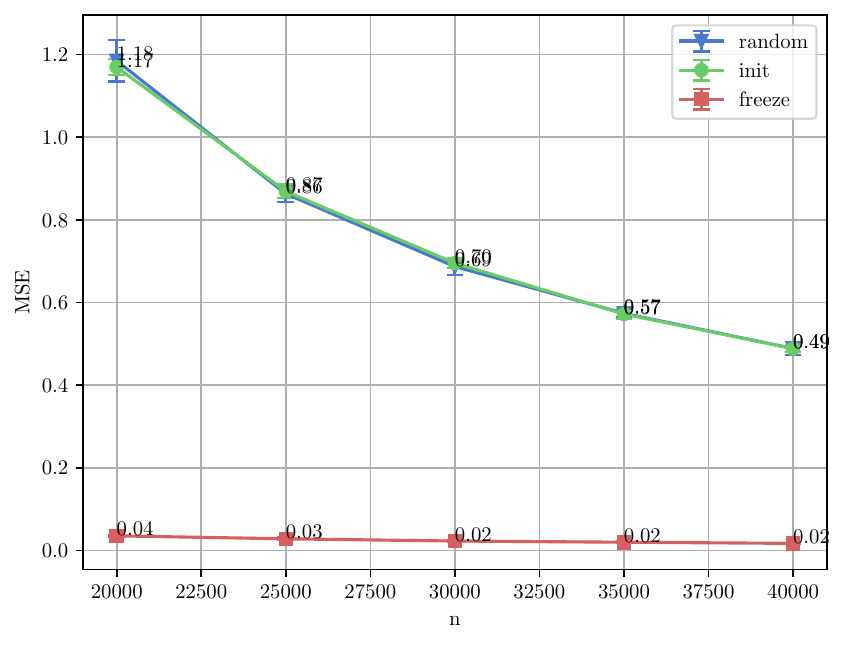}
\caption{MSEs with sufficient samples.}
\label{fig:mse_large}
\end{subfigure}
\begin{subfigure}{0.48\textwidth}
\includegraphics[width=\textwidth]{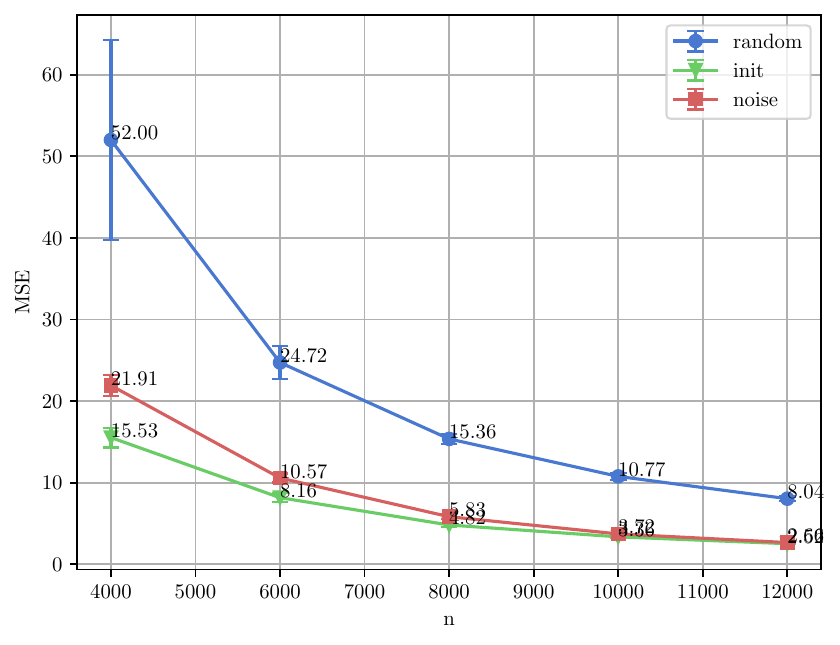}
\caption{MSEs with insufficient samples.}
\label{fig:mse_small}
\end{subfigure}
\caption{Effect of initial embedding.
We report the mean of $10$ independent random runs
and the error bars represent the standard deviation calculated from them.}
\label{fig:mse_embed}
\end{figure}

In the parametric regime (large sample size),
all MSEs show a $\frac{1}{n}$-dependency on the sample size $n$.
Moreover, the models trained with initialization achieve similar
MSEs of these with random initialization.
This suggested that initialization does not have observable effect
when the sample size is large,
as implied by the theorems if the same optimization error is achieved.
Meanwhile, by freezing the embedding, the MSEs are much lower.
This is due to the model's reduced number of parameters hence the pseudo-dimension
(see Lemma~\ref{lem:Pdim_fix}).
In the overparametrized regime,
all models fit the training data pretty well,
while the models initialized with the true embeddings, even with large noise,
achieve smaller out-of-sample MSEs compared to random initialization.

\subsection{Comparison of methods}
We further compare the performance of different models in a classification task
where the observed $(y_1, \ldots, y_n)$ is a binary vector satisfying the definition \eqref{eq:kg}.
Therefore, $\E[y_i] = \gamma(x_i)$
and $y_i \sim \mathrm{Bernoulli}(\gamma(x_i))$ are independent Bernoulli
random variables for $i = 1, \ldots, n$.
We consider two models for data generation with the embeddings and relation vectors
similar to previous definitions and an additional constant bias $b$.
The first one is a logistic model
\begin{equation*}
\gamma(x = (h, r, t)) = \sigma((\bz_h^\top, \bz_t^\top)^\top \bv_r + b)
\end{equation*}
where the link function $\sigma$ is the logistic function.
The second model is the multilayer inner product~(MIP) model
(similar to the mulilayer latent space model in Table~\ref{tab:score})
\begin{equation*}
\gamma(x = (h, r, t)) = \sigma(\bz_h^\top \bLambda_r \bz_t + b)
\end{equation*}
where $\bLambda_r = \diag(v_1, \ldots, v_d)$ is a diagonal matrix for the relation $r$
and $\sigma$ is again the logistic function.
In order to make the generation close to real data,
we choose the bias such that the generated graph is sparse.
We further assume that only positive edges (where $y_i = 1$) are observed.
All parameters (except for the biases) are independent standard Gaussian
random variables.

In these experiments, we fix $N=100$, $D=10$, and $K=3$.
For the proposed neural knowledge graph model,
we use the C-NKG with one hidden layer of $32$ units to fit the logistic generated data,
and two hidden layers of $64$ and $32$ units respectively to fit the MIP generated data.
Both TransE and C-NKG are trained by minimizing a margin-based contrastive loss
with uniform negative sampling using gradient-based optimization.
The results are shown in Figure~\ref{fig:err_comp}.
\begin{figure}[ht]
\centering
\begin{subfigure}{0.48\textwidth}
\includegraphics[width=\textwidth]{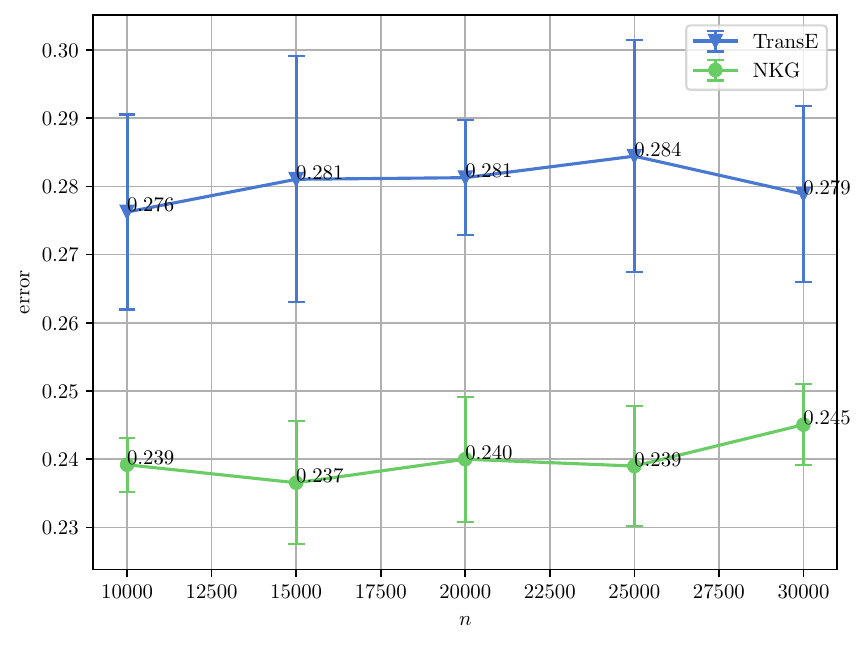}
\caption{Logistic model.}
\label{fig:err_log}
\end{subfigure}
\begin{subfigure}{0.48\textwidth}
\includegraphics[width=\textwidth]{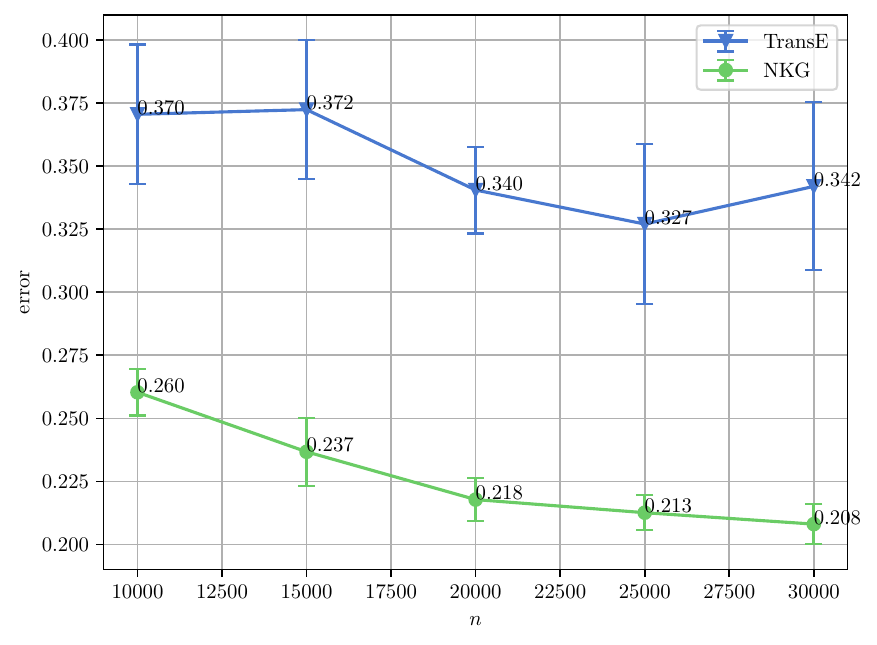}
\caption{MIP model.}
\label{fig:err_mls}
\end{subfigure}
\caption{Classification error for different methods.
The error bars are calculated from $10$ independent random runs.}
\label{fig:err_comp}
\end{figure}

We see that the NKG model outperforms TransE consistently
in both data generating processes.
The improvements become larger when the score function is more complex
(MIP vs. logistic).

\section{Medical knowledge graph learning}
\label{sec:med}
In this part,
we apply our models to a medical knowledge graph learning task
where we aim to model various types of relations among
several categories of concepts, including both codified and narrative ones. 
Our experiments suggest that using either the pre-trained language model embedding alone
or the standard knowledge graph learning methods
without the assistance from pre-trained model emebeddings have their limitations
in modeling some relation types,
while combining them as in the RENKI framework
results in much better performances.
We also show that weighting by relation types can improve both the classification errors
and AUCs for the medical knowledge graph learning.

A medical knowledge graph enhances generalizability by integrating diverse medical data sources
and linking concepts across conditions, treatments, and populations,
ensuring more comprehensive insights.
It helps mitigate biases by providing a structured,
evidence-based framework that reduces reliance on incomplete or skewed datasets.
Additionally, it minimizes hallucination in AI-driven healthcare applications
by anchoring predictions to verified medical relationships,
improving accuracy and trustworthiness.
While many knowledge bases have been curated to detail relations between different
clinical concepts, existing knowledge graph remains sparse and incomplete.
We aim to deploy the RENKI algorithm to enable the prediction of relationships
between entity pairs based on observed partial information on the links.
To this end, we curated a medical knowledge graph gathering information from several sources, 
focusing on three types of nodes:
the concept unique identifies (CUIs) from the Unified Medical Language System~(UMLS)~\citep{bodenreider2004unified},
along with two important categories of EHR codes---phecode representing
diagnosis~\citep{bastarache2021using} and RxNorm representing medications. 
The UMLS includes a large biomedical thesaurus that integrates
nearly 200 different vocabularies.
We focused on CUIs associated with phenotypes and medications.
Phecodes, originally developed to conduct phenome-wide association studies, have been used to support a wide range of EHR-based
translational studies on disease phenotypes~\citep{bastarache2021using}.
They are created by manually grouping
International Classification of Diseases~(ICD)~\citep{world1978international} 
codes to encode useful granularity for healthcare research.
Phecodes are maintained by the Center for Precision Medicine 
at Vanderbilt University Medical Center (VUMC).
RxNorm is a normalized naming system for medications produced
by the National Library of Medicine~(NLM)~\citep{nelson2011normalized}.
The abbreviation and counts of the entities are listed in Table~\ref{tab:node}.
\begin{table}[H]
\centering
\caption{Description of node types and statistics.}
\label{tab:node}
\begin{tabular}{l|l|r}
\hline
node type & source & count\\
\hline
CUI & UMLS & 113,787\\
phecode & VUMC & 1,846\\
RxNorm & NLM & 2,921\\
\hline 
\end{tabular}
\end{table}

Several types of relations are extracted from multiple sources
and grouped into a few large categories. 
The node types involved in each relation type,
and statistics of the relations are listed in Table~\ref{tab:relstat}.
More detailed descriptions of relation types and their precise sources
can be found in Table~\ref{tab:rel} in the supplementary material.

\begin{table}[ht]
\centering
\caption{Details of relation types.}
\label{tab:relstat}
\begin{tabular}{l|l|l|r}
\hline
relation & node types & source & count\\
\hline
CUI to phecode map & (CUI, phecode) & UMLS and UK biobank & 120,986\\
CUI to RxNorm map & (CUI, RxNorm) & UMLS & 16,849\\
\hline
CUI similar & (CUI, CUI) & SNOMEDCT & 22,396\\
CUI broader & (CUI, CUI) & SNOMEDCT & 348,074\\
CUI relatedness & (CUI, CUI) & SNOMEDCT + MEDRT & 82,792\\
\hline
phecode hierarchy & (phecode, phecode) & phecode & 4,403\\
phecode relatedness & (phecode, phecode) & Wikidata & 2,431\\
\hline
drug indication & (RxNorm, phecode) & DrugCentral & 4,591\\
drug side effect & (RxNorm, phecode) & SIDER & 95,846\\
\hline
TOTAL & - & - & 698,368 \\
\hline 
\end{tabular}
\end{table}

\subsection{AUC comparison}
We first compare the proposed framework to two baseline methods.
The first one is based on a domain-specific large language model
pre-trained for biomedical natural language processing
called PubMedBERT~\citep{gu2021domain}.
A later version fine-tuned using
sentence transformer~\citep{reimers2019sentence}
is used here due to its improved performance
and the inclusion of a decoder.
The cosine similarity of the embedding vectors for related concepts/codes is
used to measure the likelihood of an edge in the knowledge graph.
The second one is the TransE~\citep{bordes2013translating} algorithm
which learns the embedding of nodes from the knowledge graph
using a translation-based distance metric~\citep{mikolov2013linguistic}
and a contrastive loss.
For the proposed RENKI framework,
we also analyze two specific models: RENKI--EMB and RENKI--NKG.
RENKI--EMB includes both entity and relation embeddings,
which are input into a score function identical to
that of the TransE model shown in Table~\ref{tab:score}.
Specifically, the function is defined as
$f(x = (h, r, t)) = -\norm{\bz_h - \bz_t + \bv_r}^2$.
RENKI--NKG follows the C-NKG model outlined in Definition~\ref{def:cnkg}.
Here the ReLU network for each relation type consists of two hidden layers
with $256$ and $100$ hidden units respectively.
Both RENKI models leverage PubMedBERT to initialize embeddings for entities
based on their text descriptions,
training the embedding parameters alongside the other model parameters.

We evaluate the model performance by the AUCs (Area under the ROC Curve)
of all relation types.
Negative edges are created by replacing either head or tail node with a random
sample from the same node category.
Each experiment is repeated 10 times with random 80/20 train/test splits
(if applicable) and we report the means as well as the one standard deviations.
The results are listed in Table~\ref{tab:res}.
The last line is a weighted average of the AUCs of all relation types.
\begin{table}[ht]
\centering
\caption{Comparison of methods.
The best results are highlighted in bold.
The standard deviations calculated from $10$ independent runs are in parentheses.}
\label{tab:res}
\begin{tabular}{l|r|c|c|c}
\hline
type & PubMedBERT & TransE & RENKI--EMB & RENKI--NKG\\
\hline
CUI to phecode map & $0.940 (0.000)$ & $0.556 (0.003)$  & $0.981 (0.001)$ & $\mathbf{0.987} (0.001)$\\
CUI to RxNorm map & $0.976 (0.000)$  & $0.557 (0.007)$ & $\mathbf{0.997} (0.000)$ & $0.988 (0.004)$\\
\hline
CUI similarity  & $0.986 (0.000)$  & $0.872 (0.004)$ & $0.994 (0.000)$ & $\mathbf{0.997} (0.001)$\\
CUI broader & $0.963 (0.000))$ & $0.682 (0.001)$ & $0.991 (0.000)$ & $\mathbf{0.992} (0.001)$\\
CUI relatedness & $0.845 (0.000)$ & $0.756 (0.002)$ & $0.959 (0.002)$ & $\mathbf{0.988} (0.001)$\\
\hline
phecode similarity & $0.953 (0.001)$ & $0.621 (0.011)$ & $0.973 (0.005)$ & $\mathbf{0.991} (0.002)$\\
phecode relatedness& $0.726 (0.004)$ & $0.748 (0.017)$ & $0.922 (0.012)$ & $\mathbf{0.965} (0.007)$\\
\hline
drug indication & $0.709 (0.005)$ & $0.808 (0.009)$ & $0.917 (0.004)$ & $\mathbf{0.956} (0.005)$\\
drug side effect & $0.545 (0.001)$ & $0.822 (0.002)$ & $0.868 (0.001)$ & $\mathbf{0.897} (0.001)$\\
\hline
AVG & $0.886 (0.000)$ & $0.692 (0.001)$ & $0.968 (0.000)$ & $\mathbf{0.977} (0.001)$\\
\hline
\end{tabular}
\end{table}

The first observation is that the language model-based approach excels in capturing semantic-based relationships, 
such as CUI-to-phecode and CUI-to-RxNorm mappings, CUI and phecode similarity, and broader CUI classes.
However, its performance significantly declines in relatedness types, including CUI and phecode relatedness,
as well as drug indications and side effects.
This is because relatedness types often stem from factors beyond semantic similarity,
which the language model struggles to capture directly through embeddings.
In contrast, the knowledge graph learning method effectively captures various relatedness information,
particularly drug indications and side effects,
but struggles with more semantic-based mappings due to the sparsity of these relationships.
The RENKI--EMB model, which combines language embeddings with knowledge graphs,
significantly outperforms the two baseline methods.
Our final approach, RENKI--NKG, instead of using a predefined score function,
leverages ReLU networks to capture the functional representation of relation types,
providing greater flexibility in modeling.
This method achieves superior results,
especially for the more challenging relations such as CUI and phecode relatedness,
and drug--disease relationships.
The only case where the simpler RENKI--EMB model has an advantage is mapping CUIs
to ingredient-level RxNorms,
likely due to this relation being both highly semantic and sparse,
and hence simpler models are less prone to overfitting.

\subsection{Effect of weighting}
We next demonstrate the effect of weighting in the medical knowledge graph
learning task.
Even though our theory allows us to have noise that depends on each triple,
such information is difficult to obtain and not scalable.
Hence, we make a further simplification that the noise depends on the relation
type rather than the nodes.
It can also be observed in practice that some types of relations have smaller
noise while others have larger noise.
Specifically, the mappings and similarities can usually be identified by
semantics and have less ambiguity,
while relatedness types are often more noisy and unreliable.
Therefore, we assume that the relatedness types have a larger variance
and weight the samples according to their noise levels.
This leaves us with one tuning parameter in the model: the ratio between
the weights of low-noise relation types
(CUI to phecode and RxNorm map, CUI similarity and broader, phecode similarity)
and high-noise relation types
(CUI and phecode relatedness, drug indication and size effect).
We use $\lambda$ to denote this parameter.
Experiments are conducted on the RENKI--NKG method with varying $\lambda$.
We report the results measured in classification errors as well as AUCs
in Figure~\ref{fig:weight}.
\begin{figure}[ht]
\centering
\begin{subfigure}{0.48\textwidth}
\includegraphics[width=\textwidth]{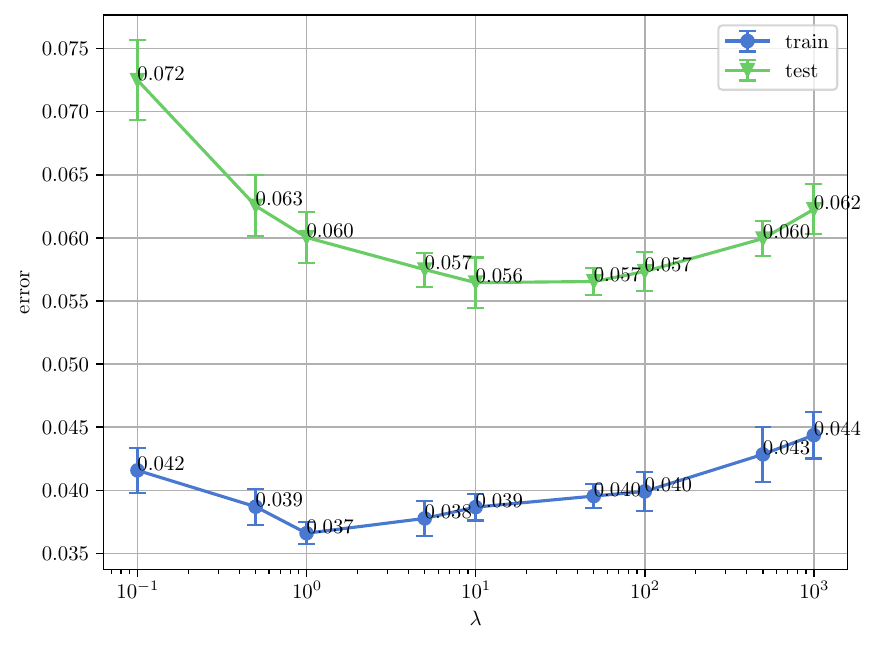}
\caption{Error for different weight ratios.}
\end{subfigure}
\begin{subfigure}{0.48\textwidth}
\includegraphics[width=\textwidth]{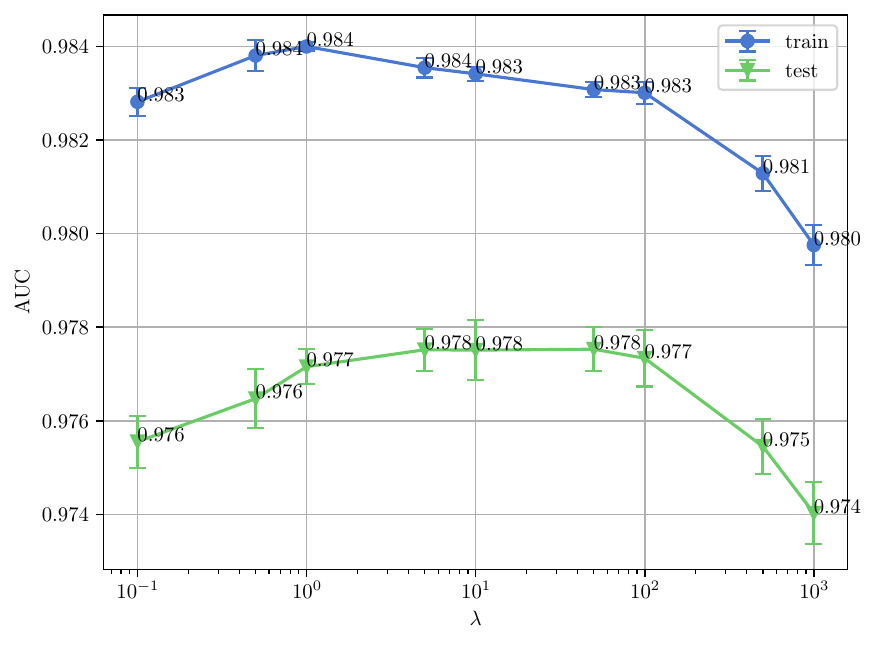}
\caption{AUC for different weight ratios.
} 
\end{subfigure}
\caption{Effect of weighting.}
\label{fig:weight}
\end{figure}

The test error is the smallest when $\lambda$ is around $10$.
For AUC,
there seems to be an optimal range of $\lambda$ from $1$ to $100$ centered around $10$
in logarithmic scale.
The AUC drops significantly outside this region.
This suggests that weighting does affect both the errors and AUCs
of medical knowledge graph learning in a nontrivial way
and our intuition about the noise in relation types is
indeed reflected in the real-world data.

\section{Discussion}
\label{sec:disc}
In this paper, we propose a knowledge graph learning framework
that combines statistical modeling and representation learning.
Parameter estimation using weighted least squares is analyzed
with nonasymptotic bounds on both in-sample and out-of-sample MSEs.
The results are established with the pseudo-dimension of the function class
hence covering a large variety of models proposed in the literature.
We further provide upper bounds on the pseudo-dimension of a neural knowledge graph
function class with ReLU activation when the embeddings are fixed and trainable.

We view the current work as a first step towards this general and intriguing question
of pre-training parameters for knowledge graph models.
It would be interesting to further study to what extent the initialization help with
the knowledge graph learning task under different model assumptions on the relationship
between the initial embeddings and the true underlying embeddings.
From an application standpoint,
the neural knowledge graph models we develop can significantly enhance biomedical knowledge
by filling in gaps where many facts on relationships are currently missing or undiscovered.
For instance, the knowledge graph constructed from phenotypic and medication data offers
valuable insights into potential new connections between drugs and diseases,
which could lead to the discovery of novel drug repurposing opportunities.
The current framework can easily accommodate additional concepts beyond phenotypes and medications
such as genes, lab tests, and environmental factors,
thereby increasing its utility across various domains.
Furthermore, while large language models (LLMs) have demonstrated impressive performance,
they often suffer from hallucinations, especially in the medical domain~\citep{pal2023med}.
Knowledge injection techniques~\citep{fu2023revisiting} have been proposed to address this issue.
Our structured knowledge graph models, which present information in a more systematic and interpretable way, 
can help LLMs mitigate hallucinations by providing a reliable foundation of factual knowledge.
This not only improves their accuracy but also enhances their application in complex biomedical tasks.

\if\JASA1
\bigskip
\begin{center}
{\large\bf SUPPLEMENTARY MATERIAL}
\end{center}

\begin{description}

\item {\bf Supplement to ``\papertitle'':} This supplement contains proofs for the lemmas,
theorems, and intermediate results.

\end{description}
\fi

\bibliographystyle{chicago}
\bibliography{ref}

\begin{thebibliography}{}

\bibitem[\protect\citeauthoryear{Anthony and Bartlett}{Anthony and Bartlett}{1999}]{anthony1999neural}
Anthony, M. and P.~L. Bartlett (1999).
\newblock {\em Neural Network Learning: Theoretical Foundations}.
\newblock Cambridge University Press.

\bibitem[\protect\citeauthoryear{Avram, Bologa, Holmes, Bocci, Wilson, Nguyen, Curpan, Halip, Bora, Yang, et~al.}{Avram et~al.}{2021}]{avram2021drugcentral}
Avram, S., C.~G. Bologa, J.~Holmes, G.~Bocci, T.~B. Wilson, D.-T. Nguyen, R.~Curpan, L.~Halip, A.~Bora, J.~J. Yang, et~al. (2021).
\newblock Drugcentral 2021 supports drug discovery and repositioning.
\newblock {\em Nucleic acids research\/}~{\em 49\/}(D1), D1160--D1169.

\bibitem[\protect\citeauthoryear{Bartlett, Harvey, Liaw, and Mehrabian}{Bartlett et~al.}{2019}]{bartlett2019nearly}
Bartlett, P.~L., N.~Harvey, C.~Liaw, and A.~Mehrabian (2019).
\newblock Nearly-tight vc-dimension and pseudodimension bounds for piecewise linear neural networks.
\newblock {\em The Journal of Machine Learning Research\/}~{\em 20\/}(1), 2285--2301.

\bibitem[\protect\citeauthoryear{Bartlett, Maiorov, and Meir}{Bartlett et~al.}{1998}]{bartlett1998almost}
Bartlett, P.~L., V.~Maiorov, and R.~Meir (1998).
\newblock Almost linear {VC}-dimension bounds for piecewise polynomial networks.
\newblock {\em Neural Computation\/}~{\em 10\/}(8), 2159--2173.

\bibitem[\protect\citeauthoryear{Bastarache}{Bastarache}{2021}]{bastarache2021using}
Bastarache, L. (2021).
\newblock Using phecodes for research with the electronic health record: from phewas to phers.
\newblock {\em Annual review of biomedical data science\/}~{\em 4\/}(1), 1--19.

\bibitem[\protect\citeauthoryear{Bodenreider}{Bodenreider}{2004}]{bodenreider2004unified}
Bodenreider, O. (2004).
\newblock The unified medical language system ({UMLS}): integrating biomedical terminology.
\newblock {\em Nucleic Acids Research\/}~{\em 32\/}(Database issue), D267--D270.

\bibitem[\protect\citeauthoryear{Bordes, Usunier, Garcia-Duran, Weston, and Yakhnenko}{Bordes et~al.}{2013}]{bordes2013translating}
Bordes, A., N.~Usunier, A.~Garcia-Duran, J.~Weston, and O.~Yakhnenko (2013).
\newblock Translating embeddings for modeling multi-relational data.
\newblock In C.~Burges, L.~Bottou, M.~Welling, Z.~Ghahramani, and K.~Weinberger (Eds.), {\em Advances in Neural Information Processing Systems}, Volume~26. Curran Associates, Inc.

\bibitem[\protect\citeauthoryear{Boucheron, Lugosi, and Massart}{Boucheron et~al.}{2013}]{boucheron2003concentration}
Boucheron, S., G.~Lugosi, and P.~Massart (2013, 02).
\newblock {\em {Concentration Inequalities: A Nonasymptotic Theory of Independence}}.
\newblock Oxford University Press.

\bibitem[\protect\citeauthoryear{Brown, Mann, Ryder, Subbiah, Kaplan, Dhariwal, Neelakantan, Shyam, Sastry, Askell, Agarwal, Herbert-Voss, Krueger, Henighan, Child, Ramesh, Ziegler, Wu, Winter, Hesse, Chen, Sigler, Litwin, Gray, Chess, Clark, Berner, McCandlish, Radford, Sutskever, and Amodei}{Brown et~al.}{2020}]{brown2020language}
Brown, T., B.~Mann, N.~Ryder, M.~Subbiah, J.~D. Kaplan, P.~Dhariwal, A.~Neelakantan, P.~Shyam, G.~Sastry, A.~Askell, S.~Agarwal, A.~Herbert-Voss, G.~Krueger, T.~Henighan, R.~Child, A.~Ramesh, D.~Ziegler, J.~Wu, C.~Winter, C.~Hesse, M.~Chen, E.~Sigler, M.~Litwin, S.~Gray, B.~Chess, J.~Clark, C.~Berner, S.~McCandlish, A.~Radford, I.~Sutskever, and D.~Amodei (2020).
\newblock Language models are few-shot learners.
\newblock In H.~Larochelle, M.~Ranzato, R.~Hadsell, M.~Balcan, and H.~Lin (Eds.), {\em Advances in Neural Information Processing Systems}, Volume~33, pp.\  1877--1901. Curran Associates, Inc.

\bibitem[\protect\citeauthoryear{Chen, Zhu, Ling, Inkpen, and Wei}{Chen et~al.}{2018}]{chen2017neural}
Chen, Q., X.~Zhu, Z.-H. Ling, D.~Inkpen, and S.~Wei (2018, July).
\newblock Neural natural language inference models enhanced with external knowledge.
\newblock In I.~Gurevych and Y.~Miyao (Eds.), {\em Proceedings of the 56th Annual Meeting of the Association for Computational Linguistics (Volume 1: Long Papers)}, Melbourne, Australia, pp.\  2406--2417. Association for Computational Linguistics.

\bibitem[\protect\citeauthoryear{Dong, Gabrilovich, Heitz, Horn, Lao, Murphy, Strohmann, Sun, and Zhang}{Dong et~al.}{2014}]{dong2014knowledge}
Dong, X., E.~Gabrilovich, G.~Heitz, W.~Horn, N.~Lao, K.~Murphy, T.~Strohmann, S.~Sun, and W.~Zhang (2014).
\newblock Knowledge vault: a web-scale approach to probabilistic knowledge fusion.
\newblock In {\em Proceedings of the 20th ACM SIGKDD International Conference on Knowledge Discovery and Data Mining}, KDD '14, New York, NY, USA, pp.\  601--610. Association for Computing Machinery.

\bibitem[\protect\citeauthoryear{Donnelly et~al.}{Donnelly et~al.}{2006}]{donnelly2006snomed}
Donnelly, K. et~al. (2006).
\newblock {SNOMED-CT}: The advanced terminology and coding system for ehealth.
\newblock {\em Studies in health technology and informatics\/}~{\em 121}, 279.

\bibitem[\protect\citeauthoryear{Fan and Gu}{Fan and Gu}{2023}]{fan2023factor}
Fan, J. and Y.~Gu (2023).
\newblock Factor augmented sparse throughput deep {ReLU} neural networks for high dimensional regression.
\newblock {\em Journal of the American Statistical Association\/}~{\em 0\/}(0), 1--15.

\bibitem[\protect\citeauthoryear{Fu, Zhang, Wang, Qiu, and Zhao}{Fu et~al.}{2023}]{fu2023revisiting}
Fu, P., Y.~Zhang, H.~Wang, W.~Qiu, and J.~Zhao (2023).
\newblock Revisiting the knowledge injection frameworks.
\newblock In {\em Proceedings of the 2023 Conference on Empirical Methods in Natural Language Processing}, pp.\  10983--10997.

\bibitem[\protect\citeauthoryear{Gu, Tinn, Cheng, Lucas, Usuyama, Liu, Naumann, Gao, and Poon}{Gu et~al.}{2021}]{gu2021domain}
Gu, Y., R.~Tinn, H.~Cheng, M.~Lucas, N.~Usuyama, X.~Liu, T.~Naumann, J.~Gao, and H.~Poon (2021).
\newblock Domain-specific language model pretraining for biomedical natural language processing.
\newblock {\em ACM Transactions on Computing for Healthcare (HEALTH)\/}~{\em 3\/}(1), 1--23.

\bibitem[\protect\citeauthoryear{Gy{\"o}rfi, Kohler, Krzyzak, Walk, et~al.}{Gy{\"o}rfi et~al.}{2002}]{gyorfi2002distribution}
Gy{\"o}rfi, L., M.~Kohler, A.~Krzyzak, H.~Walk, et~al. (2002).
\newblock {\em A Distribution-Free Theory of Nonparametric Regression}, Volume~1 of {\em Springer Series in Statistics}.
\newblock Springer New York, NY.

\bibitem[\protect\citeauthoryear{Huang, Altosaar, and Ranganath}{Huang et~al.}{2019}]{huang2019clinicalbert}
Huang, K., J.~Altosaar, and R.~Ranganath (2019).
\newblock {ClinicalBERT}: Modeling clinical notes and predicting hospital readmission.
\newblock {\em arXiv preprint arXiv:1904.05342\/}.

\bibitem[\protect\citeauthoryear{Ji, Pan, Cambria, Marttinen, and Philip}{Ji et~al.}{2021}]{ji2021survey}
Ji, S., S.~Pan, E.~Cambria, P.~Marttinen, and S.~Y. Philip (2021).
\newblock A survey on knowledge graphs: Representation, acquisition, and applications.
\newblock {\em IEEE transactions on neural networks and learning systems\/}~{\em 33\/}(2), 494--514.

\bibitem[\protect\citeauthoryear{Kuhn, Letunic, Jensen, and Bork}{Kuhn et~al.}{2016}]{kuhn2016sider}
Kuhn, M., I.~Letunic, L.~J. Jensen, and P.~Bork (2016).
\newblock The sider database of drugs and side effects.
\newblock {\em Nucleic acids research\/}~{\em 44\/}(D1), D1075--D1079.

\bibitem[\protect\citeauthoryear{Lee, Yoon, Kim, Kim, Kim, So, and Kang}{Lee et~al.}{2020}]{lee2020biobert}
Lee, J., W.~Yoon, S.~Kim, D.~Kim, S.~Kim, C.~H. So, and J.~Kang (2020).
\newblock {BioBERT}: a pre-trained biomedical language representation model for biomedical text mining.
\newblock {\em Bioinformatics\/}~{\em 36\/}(4), 1234--1240.

\bibitem[\protect\citeauthoryear{Li, Xu, and Zhu}{Li et~al.}{2023}]{li2023statistical}
Li, J., G.~Xu, and J.~Zhu (2023).
\newblock Statistical inference on latent space models for network data.
\newblock {\em arXiv preprint arXiv:2312.06605\/}.

\bibitem[\protect\citeauthoryear{Ma, Ma, and Yuan}{Ma et~al.}{2020}]{ma2020universal}
Ma, Z., Z.~Ma, and H.~Yuan (2020).
\newblock Universal latent space model fitting for large networks with edge covariates.
\newblock {\em Journal of Machine Learning Research\/}~{\em 21\/}(4), 1--67.

\bibitem[\protect\citeauthoryear{MacDonald, Levina, and Zhu}{MacDonald et~al.}{2022}]{macdonald2022latent}
MacDonald, P.~W., E.~Levina, and J.~Zhu (2022).
\newblock Latent space models for multiplex networks with shared structure.
\newblock {\em Biometrika\/}~{\em 109\/}(3), 683--706.

\bibitem[\protect\citeauthoryear{Mikolov, Yih, and Zweig}{Mikolov et~al.}{2013}]{mikolov2013linguistic}
Mikolov, T., W.-t. Yih, and G.~Zweig (2013).
\newblock Linguistic regularities in continuous space word representations.
\newblock In {\em Proceedings of the 2013 conference of the north american chapter of the association for computational linguistics: Human language technologies}, pp.\  746--751.

\bibitem[\protect\citeauthoryear{Nelson, Zeng, Kilbourne, Powell, and Moore}{Nelson et~al.}{2011}]{nelson2011normalized}
Nelson, S.~J., K.~Zeng, J.~Kilbourne, T.~Powell, and R.~Moore (2011).
\newblock Normalized names for clinical drugs: Rxnorm at 6 years.
\newblock {\em Journal of the American Medical Informatics Association\/}~{\em 18\/}(4), 441--448.

\bibitem[\protect\citeauthoryear{Nickel, Murphy, Tresp, and Gabrilovich}{Nickel et~al.}{2015}]{nickel2015review}
Nickel, M., K.~Murphy, V.~Tresp, and E.~Gabrilovich (2015).
\newblock A review of relational machine learning for knowledge graphs.
\newblock {\em Proceedings of the IEEE\/}~{\em 104\/}(1), 11--33.

\bibitem[\protect\citeauthoryear{Pal, Umapathi, and Sankarasubbu}{Pal et~al.}{2023}]{pal2023med}
Pal, A., L.~K. Umapathi, and M.~Sankarasubbu (2023).
\newblock Med-halt: Medical domain hallucination test for large language models.
\newblock In {\em Proceedings of the 27th Conference on Computational Natural Language Learning (CoNLL)}, pp.\  314--334.

\bibitem[\protect\citeauthoryear{Paul and Chen}{Paul and Chen}{2020}]{paul2020spectral}
Paul, S. and Y.~Chen (2020).
\newblock Spectral and matrix factorization methods for consistent community detection in multi-layer networks.
\newblock {\em The Annals of Statistics\/}~{\em 48\/}(1), 230--250.

\bibitem[\protect\citeauthoryear{Peter D~Hoff and Handcock}{Peter D~Hoff and Handcock}{2002}]{hoff2002latent}
Peter D~Hoff, A. E.~R. and M.~S. Handcock (2002).
\newblock Latent space approaches to social network analysis.
\newblock {\em Journal of the American Statistical Association\/}~{\em 97\/}(460), 1090--1098.

\bibitem[\protect\citeauthoryear{Reimers and Gurevych}{Reimers and Gurevych}{2019}]{reimers2019sentence}
Reimers, N. and I.~Gurevych (2019).
\newblock Sentence-{BERT}: Sentence embeddings using siamese {BERT}-networks.
\newblock In {\em Proceedings of the 2019 Conference on Empirical Methods in Natural Language Processing and the 9th International Joint Conference on Natural Language Processing (EMNLP-IJCNLP)}, pp.\  3982--3992.

\bibitem[\protect\citeauthoryear{Robinson, K{\"o}hler, Bauer, Seelow, Horn, and Mundlos}{Robinson et~al.}{2008}]{robinson2008human}
Robinson, P.~N., S.~K{\"o}hler, S.~Bauer, D.~Seelow, D.~Horn, and S.~Mundlos (2008).
\newblock The human phenotype ontology: a tool for annotating and analyzing human hereditary disease.
\newblock {\em The American Journal of Human Genetics\/}~{\em 83\/}(5), 610--615.

\bibitem[\protect\citeauthoryear{Schmidt-Hieber}{Schmidt-Hieber}{2020}]{schmidt2020nonparametric}
Schmidt-Hieber, J. (2020).
\newblock Nonparametric regression using deep neural networks with {ReLU} activation function.
\newblock {\em The Annals of Statistics\/}~{\em 48\/}(4), 1875--1897.

\bibitem[\protect\citeauthoryear{Tang, Sussman, and Priebe}{Tang et~al.}{2013}]{tang2013universally}
Tang, M., D.~L. Sussman, and C.~E. Priebe (2013).
\newblock Universally consistent vertex classification for latent positions graphs.
\newblock {\em The Annals of Statistics\/}~{\em 41\/}(3), 1406--1430.

\bibitem[\protect\citeauthoryear{Wainwright}{Wainwright}{2019}]{wainwright2019high}
Wainwright, M.~J. (2019).
\newblock {\em High-Dimensional Statistics: A Non-Asymptotic Viewpoint}.
\newblock Cambridge Series in Statistical and Probabilistic Mathematics. Cambridge University Press.

\bibitem[\protect\citeauthoryear{{World Health Organization}}{{World Health Organization}}{1978}]{world1978international}
{World Health Organization} (1978).
\newblock {\em International classification of diseases:[9th] ninth revision, basic tabulation list with alphabetic index}.
\newblock World Health Organization.

\bibitem[\protect\citeauthoryear{Yao and Van~Durme}{Yao and Van~Durme}{2014}]{yao2014information}
Yao, X. and B.~Van~Durme (2014).
\newblock Information extraction over structured data: Question answering with freebase.
\newblock In {\em Proceedings of the 52nd annual meeting of the association for computational linguistics (volume 1: long papers)}, pp.\  956--966.

\bibitem[\protect\citeauthoryear{Yarotsky}{Yarotsky}{2017}]{yarotsky2017error}
Yarotsky, D. (2017).
\newblock Error bounds for approximations with deep {ReLU} networks.
\newblock {\em Neural networks\/}~{\em 94}, 103--114.

\bibitem[\protect\citeauthoryear{Zhang, Yuan, Lian, Xie, and Ma}{Zhang et~al.}{2016}]{zhang2016collaborative}
Zhang, F., N.~J. Yuan, D.~Lian, X.~Xie, and W.-Y. Ma (2016).
\newblock Collaborative knowledge base embedding for recommender systems.
\newblock In {\em Proceedings of the 22nd ACM SIGKDD international conference on knowledge discovery and data mining}, pp.\  353--362.

\bibitem[\protect\citeauthoryear{Zhang, Xue, and Zhu}{Zhang et~al.}{2020}]{zhang2020flexible}
Zhang, X., S.~Xue, and J.~Zhu (2020, 13--18 Jul).
\newblock A flexible latent space model for multilayer networks.
\newblock In H.~D. III and A.~Singh (Eds.), {\em Proceedings of the 37th International Conference on Machine Learning}, Volume 119 of {\em Proceedings of Machine Learning Research}, pp.\  11288--11297. PMLR.

\end{thebibliography}

\if\JASA1
\newpage

  \bigskip
  \bigskip
  \bigskip
  \begin{center}
    {\LARGE\bf Supplement to ``\papertitle''}
  \end{center}
  \setcounter{page}{1}
  \medskip
\fi

\appendix

\section{Details of relations and sources}
\begin{longtable}{l|p{4in}}
\caption{Descriptions of relation types.}
\label{tab:rel}\\
\hline
relation & description\\
\hline
\textbf{CUI to phecode map} &
The mapping is created by combining UMLS CUI to SNOMEDCT and read code mapping,
UK biobank SNOMEDCT and read code to ICD10 code mapping,
and the ICD10 to phecode mapping.
A further exact string matching is performed to remove redundancies.\\
\hline
\textbf{CUI to RxNorm map} &
CUIs belonging to the medications are mapped to
their ingredient-level RxNorms according to UMLS.\\
\hline
\textbf{CUI similarity} &
A group of relations from SNOMEDCT indicating some equivalence between the 
concepts such as
``same\_as'', ``has\_alternative'', and ``possibly\_equivalent\_to''.\\
\hline
\textbf{CUI broader} &
Pairs of concepts such that one represents a broader class than the other
or includes the other as a special case.\\
\hline
\textbf{CUI relatedness} &
A collection of selected and assorted relations
from SNOMEDCT and Medication Reference Terminology~(MED-RT).
The head and tail nodes are roughly ordered according to their ``causal''
relation.
The tail node usually is ``caused by'' or ``happens after'' the head node.\\
\hline
\textbf{Phecode similarity} &
Similarity between phecodes by their hierarchical structure.
Phecodes under the same integer category is treated as similar.\\
\hline
\textbf{Phecode relatedness} &
The relatedness structure between phecodes
including ``cause'', ``symptom'', ``complication'', ``risk factor'',
and ``differential diagnosis''.\\
\hline
\textbf{Drug indication} &
Drug--disease interaction extracted from DrugCentral 2021~\citep{avram2021drugcentral}
and mapped to RxNorms and phecodes.\\
\hline
\textbf{Drug side effect} &
Drug adverse reactions~(ADRs) extracted from Drug Side Effect Resource~(SIDER)~\citep{kuhn2016sider} and mapped to RxNorms and phecodes.\\
\hline
\end{longtable}

\section{Proof of Theorem~\ref{thm:is_mse}}
For any function $f \in \cH$,
by construction \eqref{eq:approx_sol},
\begin{equation*}
\frac{1}{n} \sum_{i=1}^{n} w_i (\fh(x_i) - y_i)^2
\le \frac{1}{n} \sum_{i=1}^{n} w_i (f(x_i) - y_i)^2
+ \delta_\opt.
\end{equation*}
Using the model assumption \eqref{eq:kg}
and recalling the definition of $\norm{\cdot}_{\bw, n}$,
we have that
\begin{equation} \label{eq:fh-gamma}
\norm{\fh - \gamma}_{\bw, n}^2
\le \norm{f - \gamma}_{\bw, n}^2
+ \frac{2}{n} \sum_{i=1}^{n} \varepsilon_i w_i (\fh(x_i) - f(x_i))
+ \delta_\opt.
\end{equation}
We utilize the following lemma
that provides a high-probability bound on the second term in the above display
uniformly over all functions in $\cH$.
The proof of the lemma is deferred to Section~\ref{sec:emp_pro_proof}. 
\begin{lemma} \label{lem:emp_pro}
Fix $x_1, \ldots, x_n$
and let $\varepsilon_1, \ldots, \varepsilon_n$ be i.i.d.\
sub-Gaussian random variables.
Suppose $\cH$ is a function class with peudo-dimension
$p \coloneq \Pdim(\cH) \le n$ and let $\epsilon > 0$.
Then for any $t > 0$,
\begin{equation*}
\begin{split}
&\P\biggl\{\exists f, f \in \cF: \frac{1}{n} \sum_{i=1}^{n} \varepsilon_i w_i (f(x_i) - f'(x_i))
\ge t (\norm{f - f'}_{\bw, n} + 4\epsilon\sqrt{\wb})\biggr\}\\
&\qquad\le 2\biggl(\frac{enB}{\epsilon p}\biggr)^{2p}
\exp\biggl(-\frac{nt^2}{2\sigma_\ttm^2}\biggr).
\end{split}
\end{equation*}
\end{lemma}

Using Lemma~\ref{lem:emp_pro} with $t = \sqrt{\frac{6p \sigma_{\ttm}^2}{n}
\log \frac{enB}{\epsilon p}}$,
for a positive $\epsilon \le B$ that will be specified later,
we obtain with probability at least
$1 - 2(\frac{p}{en})^{p}$,
for any $f, f' \in \cF$,
\begin{equation} \label{eq:f-fp}
\frac{1}{n} \sum_{i=1}^{n} \varepsilon_i w_i (f(x_i) - f'(x_i))
\le (\norm{f - f'}_{\bw, n} + 4\epsilon \sqrt{\wb})
\sqrt{\frac{6p \sigma_{\ttm}^2}{n} \log \frac{enB}{\epsilon p}}.
\end{equation}

Plugging \eqref{eq:f-fp} into \eqref{eq:fh-gamma} gives
\begin{equation*}
\begin{split}
\norm{\fh - \gamma}_{\bw, n}^2
&\le \norm{f - \gamma}_{\bw, n}^2
+ 2 (\norm{\fh - f}_{\bw, n} + 4\epsilon \sqrt{\wb})
\sqrt{\frac{6p \sigma_{\ttm}^2}{n} \log \frac{enB}{\epsilon p}}
+ \delta_\opt\\
&\le \norm{f - \gamma}_{\bw, n}^2
+ \frac{1}{4}\norm{\fh - f}_{\bw, n}^2
+ \frac{24 p \sigma_{\ttm}^2}{n} \log \frac{enB}{\epsilon p}\\
&\phantom{{}\le{}}+ 8\epsilon\sqrt{\frac{6 p \wb \sigma_{\ttm}^2}{n}
\log \frac{enB}{\epsilon p}}
+ \delta_\opt.
\end{split}
\end{equation*}
Since by triangle inequality,
\begin{equation} \label{eq:fh-f}
\norm{\fh - f}_{\bw, n}^2
\le (\norm{\fh - \gamma}_{\bw, n} + \norm{f - \gamma}_{\bw, n})^2
\le 2\norm{\fh - \gamma}_{\bw, n}^2 + 2\norm{f - \gamma}_{\bw, n}^2,
\end{equation}
we obtain that
\begin{equation*}
\begin{split}
\norm{\fh - \gamma}_{\bw, n}^2
&\le 3\norm{f - \gamma}_{\bw, n}^2
+ \frac{48 p\sigma_{\ttm}^2}{n} \log \frac{enB}{\epsilon p}
+ 16\epsilon \sqrt{\frac{6p \wb \sigma_{\ttm}^2}{n} \log \frac{enB}{\epsilon p}}
+ 2\delta_\opt\\
&\le 3\norm{f - \gamma}_{\bw, n}^2
+ \frac{24p(2 + \wb)\sigma_{\ttm}^2}{n} \log \frac{enB}{\epsilon p}
+ 16\epsilon^2 + 2\delta_\opt.
\end{split}
\end{equation*}
If we take $\epsilon = B \sqrt{\frac{p}{en} \log \frac{en}{p}}$,
then $\epsilon \le B$ since $\log (ex)/x \le 1$ for all $x > 0$
and $\epsilon \ge B \sqrt{\frac{p}{en}}$ since $n \ge p$.
Therefore, the above display yields
\begin{equation}
\label{eq:fh-f_bound}
\norm{\fh - \gamma}_{\bw, n}^2
\le 3\norm{f - \gamma}_{\bw, n}^2
+ \frac{(36(2 + \wb)\sigma_{\ttm}^2 + 16B^2)p}{n}
\log \frac{en}{p} + 2\delta_\opt.
\end{equation}
Since it holds for all $f \in \cH$,
the theorem follows from taking the infimum
in the above display.

If the variances $\sigma_i^2$'s are known,
one may optimize the weight function $w$ to obtain the best rate.
As discussed previously,
without loss of generality,
we may fix $\wb = 1$.
In particular, we aim to solve the following minimax optimization:
\begin{equation} 
\min_{\frac{1}{n}\sum_{i=1}^n w_i = 1}  \max_{i} w_i \sigma_i^2.
\label{eq:minimax}
\end{equation}
The problem in \eqref{eq:minimax} is equivalent to the linear program
\begin{equation*}
\begin{split}
\min_{\bw} &\quad u\\
\st &\quad \sigma_i^2 w_i \le u \quad \forall i \in [n],\\
&\quad \frac{1}{n} \sum_{i=1}^{n} w_i = 1.
\end{split}
\end{equation*}
Since the optimum is achieved at the vertex of the simplex,
the optimal solution is given by $\sigma_1^2 w_1 = \cdots = \sigma_n^2 w_n$,
which gives $w_i \propto 1/\sigma_i^2$.
Therefore, the infinum of the rate is obtained by choosing
$w_i = \sigma_\ttH^2/\sigma_i^2$.

\section{Proof of Lemma~\ref{lem:emp_pro}}
\label{sec:emp_pro_proof}
Recall that by Assumption~\ref{assump:subGaussian} given $x_i$,
$\varepsilon_i$ is a  independent sub-Gaussian random variables
with variance $\sigma_i^2$ for all $i = 1, \ldots, n$.
Hoeffding's inequality
(see, e.g., \citep[Proposition~2.1]{wainwright2019high}) implies that
for any fixed pair of $f, f' \in \cH$,
and for all $t \ge 0$,
\begin{equation} \label{eq:hoeffding}
\P\biggl\{\frac{1}{n} \sum_{i=1}^{n} \varepsilon_i w_i (f(x_i) - f'(x_i))
\ge t\biggr\}
\le \exp\biggl(
-\frac{n^2t^2}{2 \sum_{i=1}^{n} \sigma_i^2 w_i^2 (f(x_i) - f'(x_i))^2}\biggr).
\end{equation}
Recall the definition $\sigma_{\ttm}^2 \coloneq \max_{i} w_i \sigma_i^2$.
Since
\begin{equation*}
\sum_{i=1}^{n} \sigma_i^2 w_i^2 (f(x_i) - f'(x_i))^2
\le (\max_{i} w_i \sigma_i^2) \sum_{i=1}^{n} w_i (f(x_i) - f'(x_i))^2
= \sigma_{\ttm}^2 n \norm{f - f'}_{\bw, n}^2,
\end{equation*}
The previous display implies
\begin{equation*}
\P\biggl\{\frac{1}{n} \sum_{i=1}^{n} \varepsilon_i w_i (f(x_i) - f'(x_i))
\ge t\norm{f - f'}_{\bw, n}\biggr\}
\le \exp\biggl(-\frac{nt^2}{2\sigma_{\ttm}^2}\biggr).
\end{equation*}
Applying Hoeffding's inequality to
$\frac{1}{n} \sum_{i=1}^{n} \varepsilon_i w_i$
and recalling $\wb \coloneq \frac{1}{n}\sum_{i=1}^{n} w_i$,
we also have
\begin{equation}
\label{eq:hoeffding_w}
\P\biggl\{\frac{1}{n} \sum_{i=1}^{n} \varepsilon_i w_i \ge t\biggr\}
\le \exp\biggl(-\frac{n^2t^2}{2 \sum_{i=1}^{n} \sigma_i^2 w_i^2}\biggr)
\le \exp\biggl(-\frac{nt^2}{2 \wb \sigma_{\ttm}^2}\biggr).
\end{equation}

Let $\cF^\epsilon$ be an $L_\infty$-cover of $\cF$ on $\cX$, i.e.,
$\forall f \in \cF$, there is an $f_\epsilon \in \cF^\epsilon$ such that
\begin{equation*}
\max_i \abs{f(x_i) - f_\epsilon(x_i)} \le \epsilon.
\end{equation*}
Then we have
\begin{equation*}
\norm{f - f_\epsilon}_{\bw, n}
= \sqrt{\frac{1}{n}\sum_{i=1}^n w_i (f(x_i) - f_\epsilon(x_i))^2}
\le \epsilon \sqrt{\frac{1}{n}\sum_{i=1}^n w_i} = \epsilon \sqrt{\wb}.
\end{equation*}
By \citep[Theorem~12.2]{anthony1999neural}, for $n \ge p$,
the size of $\cF^\epsilon$, or the uniform covering number,
\begin{equation*}
N_{\infty}(\epsilon, \cF, n)
\le \biggl(\frac{enB}{\epsilon p}\biggr)^{p}.
\end{equation*}

Using the covering $\cF^\epsilon$
and applying a union bound to \eqref{eq:hoeffding},
we have that
\begin{equation}
\label{eq:hoeffding_f}
\begin{split}
&\P\biggl\{
\exists f_\epsilon, f_\epsilon' \in \cF^\epsilon:
\frac{1}{n} \sum_{i=1}^{n} \varepsilon_i w_i
(f_\epsilon(x_i) - f_\epsilon'(x_i))
\ge t \norm{f_\epsilon - f_\epsilon'}\biggr\}\\
&\qquad\le N_{\infty}(\epsilon, \cF, n)^2
\exp\biggl(-\frac{nt^2}{2\sigma_\ttm^2}\biggr)
\le \biggl(\frac{enB}{\epsilon p}\biggr)^{2p}
\exp\biggl(-\frac{nt^2}{2\sigma_\ttm^2}\biggr).
\end{split}
\end{equation}
Therefore,
\begin{equation*}
\begin{split}
&\P\biggl\{\exists f, f \in \cF: \frac{1}{n} \sum_{i=1}^{n} \varepsilon_i w_i (f(x_i) - f'(x_i))
\ge t (\norm{f - f'}_{\bw, n} + 4\epsilon\sqrt{\wb})\biggr\}\\
&\qquad= \P\biggl\{\exists f, f' \in \cF:
\frac{1}{n} \sum_{i=1}^{n} \varepsilon_i w_i (f(x_i) - f_\epsilon(x_i))
+ \frac{1}{n} \sum_{i=1}^{n} \varepsilon_i w_i (f_\epsilon'(x_i) - f'(x_i))\\
&\phantom{\qquad{}= \P\biggl\{\exists f, f \in \cF:{}}
+ \frac{1}{n} \sum_{i=1}^{n} \varepsilon_i w_i (f_\epsilon(x_i) - f_\epsilon'(x_i))
\ge t (\norm{f - f'}_{\bw, n} + 4\epsilon\sqrt{\wb})\biggr\}\\
&\qquad\le \P\biggl\{\exists f, f' \in \cF:
\frac{1}{n} \sum_{i=1}^{n} \varepsilon_i w_i
(f_\epsilon(x_i) - f_\epsilon'(x_i))\\
&\phantom{\qquad{}= \P\biggl\{\exists f, f \in \cF:{}}
+ \frac{1}{n} \sum_{i=1}^{n} \varepsilon_i w_i \abs{f(x_i) - f_\epsilon(x_i)}
+ \frac{1}{n} \sum_{i=1}^{n} \varepsilon_i w_i
\abs{f_\epsilon'(x_i) - f'(x_i)}\\
&\phantom{\qquad{}= \P\biggl\{\exists f, f \in \cF:{}}\qquad
\ge t ( \norm{f_\epsilon - f_\epsilon'}_{\bw, n}
- \norm{f - f_\epsilon}_{\bw, n} - \norm{f' - f_\epsilon'}_{\bw, n}
+ 4\epsilon\sqrt{\wb})\biggr\}\\
&\qquad\le \P\biggl\{\exists f_\epsilon, f_\epsilon' \in \cF^\epsilon:
\frac{1}{n}\sum_{i=1}^{n} \varepsilon_i w_i (f_\epsilon(x_i) - f_\epsilon'(x_i))
+ \frac{2\epsilon}{n} \sum_{i=1}^{n} \varepsilon_i w_i\\
&\phantom{\qquad{}\le \P\biggl\{
\exists f_\epsilon, f_\epsilon' \in \cF^\epsilon:{}}\qquad
\ge t (\norm{f_\epsilon - f_\epsilon'}_{\bw, n} + 2\epsilon\sqrt{\wb})\biggr\}\\
&\qquad\le \P\biggl\{
\exists f_\epsilon, f_\epsilon' \in \cF^\epsilon:
\frac{1}{n} \sum_{i=1}^{n} \varepsilon_i w_i
(f_\epsilon(x_i) - f_\epsilon'(x_i))
\ge t \norm{f_\epsilon - f_\epsilon'}_{\bw, n}\biggr\}\\
&\phantom{\qquad{}\le{}}
+ \P\biggl\{\frac{1}{n} \sum_{i=1}^{n} \varepsilon_i w_i \ge t \sqrt{\wb}\biggr\}.
\end{split}
\end{equation*}

Plugging \eqref{eq:hoeffding_w} and \eqref{eq:hoeffding_f}
into the above display,
we obtain
\begin{equation*}
\begin{split}
&\P\biggl\{\exists f, f \in \cF: \frac{1}{n} \sum_{i=1}^{n} \varepsilon_i w_i (f(x_i) - f'(x_i))
\ge t (\norm{f - f'}_{\bw, n} + 4\epsilon\sqrt{\wb})\biggr\}\\
&\qquad \le \biggl(\frac{enB}{\epsilon p}\biggr)^{2p}
\exp\biggl(-\frac{nt^2}{2\sigma_\ttm^2}\biggr)
+ \exp\biggl(-\frac{nt^2}{2\sigma_\ttm^2}\biggr)
\le 2\biggl(\frac{enB}{\epsilon p}\biggr)^{2p}
\exp\biggl(-\frac{nt^2}{2\sigma_\ttm^2}\biggr).
\end{split}
\end{equation*}
The lemma is hence proved.

\section{Proof of Theorem~\ref{thm:os_mse}}
The key step of the proof is the following result for empirical processes.
\begin{lemma} \label{lem:cct}
Let $\cG$ be a set of functions $g: \cX \to [0, B]$
for a constant $0 \le B < \infty$.
Suppose $Z, Z_1, \ldots, Z_n \in \cX$ are i.i.d.\ random variables.
Let $w_\infty \coloneq \sup_{x \in \cX} w(x)$.
For any $\alpha > 0$ and $0 < \epsilon < 1$,
we have
\begin{equation*}
\P\biggl\{
\sup_{g \in \cG} \frac{\abs{\frac{1}{n} \sum_{i=1}^{n} w(Z_i)g(Z_i)
- \E\{w(Z) g(Z)\}}}
{w_\infty \alpha + \frac{1}{n} \sum_{i=1}^{n} w(Z_i)g(Z_i) + \E\{w(Z)g(Z)\}}
\ge \epsilon\biggr\}
\le 8 N_\infty\biggl(\frac{\epsilon\alpha}{17}, \cG, n\biggr)
\exp\biggl(-\frac{n\epsilon^2\alpha}{68 B}\biggr)
\end{equation*}
where $N_\infty(\epsilon, \cG, n)$ is the size of an $L_\infty$-cover of
$\cG$ on $\cX$.
\end{lemma}
The proof of the lemma is deferred to Section~\ref{sec:cct_proof}.
Here we use Lemma~\ref{lem:cct} to complete the proof of 
Theorem~\ref{thm:os_mse}.

Let
\begin{equation*}
g(x) \coloneq \abs{f(x) - \gamma(x)}^2.
\end{equation*}
Then by definition,
\begin{equation*}
\sup_{x \in \cX} g(x) \le 4 B^2.
\end{equation*}
By \citep[Theorem~12.2]{anthony1999neural}, for $n \ge p$,
\begin{equation*}
N_\infty\biggl(\frac{\epsilon\alpha}{17}, \cG, n\biggr)
\le \biggl(\frac{68 e n B^2}{\epsilon \alpha p}\biggr)^{p}.
\end{equation*}

Using Lemma~\ref{lem:cct} with $\epsilon = \frac{1}{2}$,
we have that under certain conditions which we will specify later,
\begin{equation*}
\norm{\fh - \gamma}_{w, \PP}^2
\le 3\norm{\fh - \gamma}_{\bw,n}^2 + w_\infty \alpha
\end{equation*}
and
\begin{equation*}
\norm{f - \gamma}_{\bw,n}^2
\le 3\norm{f - \gamma}_{w, \PP}^2 + w_\infty \alpha
\end{equation*}
for all $f \in \cF$.

Hence, by \eqref{eq:fh-f_bound} in the proof of Theorem~\ref{thm:is_mse},
for all $f \in \cF$,
\begin{equation*}
\begin{split}
\norm{\fh - \gamma}_{w, \PP}^2
&\le 3\biggl( 3\norm{f - \gamma}_{\bw, n}^2
+ \frac{(36(2 + \wb) \sigma_{\ttm}^2 + 16 B^2)p}{n} \log \frac{en}{p}
+ 2\delta_\opt \biggr)
+ w_\infty \alpha\\
&\le 27 \norm{f - \gamma}_{w, \PP}^2
+ \frac{(108(2 + \wb) \sigma_{\ttm}^2 + 48 B^2)p}{n} \log \frac{en}{p}
+ 10 w_\infty \alpha + 6\delta_\opt.
\end{split}
\end{equation*}
By choosing
\begin{equation*}
\alpha = \frac{2176 B^2 p}{n} \log \frac{e n}{p},
\end{equation*}
we have with probability at least $1 - 10(\frac{p}{en})^p$,
\begin{equation*}
\begin{split}
\norm{\fh - \gamma}_{w, \PP}^2
&\le 27\norm{f - \gamma}_{w, \PP}^2
+ \frac{(108(2 + \wb) \sigma_{\ttm}^2 + 21760 w_\infty B^2 + 48 B^2)p}{n}
\log \frac{en}{p} + 6\delta_\opt\\
\end{split}
\end{equation*}
The claim directly follows by taking the infimum over all $f \in \cF$.

For the out-of-sample MSE,
optimizing for the weight function is less straightforward.
The best bound one can hope for is all three quantities involving the weights, 
$\wb$, $\sigma_\ttm^2$, and $w_\infty$ are constants (may involve $B$).
We first try to gain some intuition from the proof of Theorem~\ref{thm:is_mse}
by assuming $w_\infty \approx \max_{i \in [n]} w(Z_i)$.
Then,
one may optimize $w_1, \ldots, w_n$ for the upper bound given
$\sigma_1^2, \ldots, \sigma_n^2$.
To achieve this,
consider the following optimization problem for $a_1, \ldots, a_n, b \ge 0$,
\begin{equation*}
\min_{\frac{1}{n} \sum_{i=1}^n w_i = 1}  (\max_{i} a_i w_i + b \max_{i} w_i).
\label{eq:mmab}
\end{equation*}
The above optimization is equivalent to the linear program
\begin{equation*}
\begin{split}
\min_{\bw} &\quad u + v\\
\st &\quad a_i w_i \le u, b w_i \le v, \forall i \in [n],\\
&\quad \frac{1}{n} \sum_{i=1}^{n} w_i = 1.
\end{split}
\end{equation*}
Hence, $w_i \propto 1/\wt{a}_i$ where $\wt{a}_i \coloneq \max\{a_i, b\}$ is a solution to the above program,
and the objective becomes $2\wt{a}_H$ where
$\wt{a}_H \coloneq \bigl(\frac{1}{n} \sum_{i=1}^{n} \frac{1}{\wt{a}_i}\bigr)^{-1}$
is the harmonic mean of $\wt{a}_1, \ldots, \wt{a}_n$.
For the expected MSE, controlling $\wb$ would be difficult:
It would be equivalent to control the expectation $\E[w(X)]$.
This will inevitably require additional assumptions on $\sigma$.
One may alternatively control the expectation by the sample mean $\wb$
using Hoeffding's inequality,
however, it will result in a rate of $1/\sqrt{n}$ that is slower than
the rate of convergence for the MSEs.
Therefore,
we opt to control $w_\infty = 1$ instead
and therefore obtaining $\wb \le w_\infty = 1$.
This is achieved by setting $w(x) = B^2 / \max\{\sigma(x)^2, B^2\}$
as in the theorem.

\section{Proof of Lemma~\ref{lem:cct}}
\label{sec:cct_proof}
The proof follows a few steps similar to
that of~\cite[Theorem~11.6]{gyorfi2002distribution}.

\noindent\textsc{Step 1.} Draw ghost samples.

Let $Z_1', \ldots, Z_n'$ are i.i.d.\ random variables with the same distribution
of $Z$.
Then,
\begin{equation*}
\biggl\lvert \frac{1}{n} \sum_{i=1}^{n} w(Z_i)g(Z_i)
- \E\{w(Z) g(Z)\} \biggr\rvert
> \epsilon\biggl(w_\infty \alpha + \frac{1}{n} \sum_{i=1}^{n} w(Z_i)g(Z_i)
+ \E\{w(Z)g(Z)\}\biggr)
\end{equation*}
and
\begin{equation*}
\biggl\lvert \frac{1}{n} \sum_{i=1}^{n} w(Z_i')g(Z_i')
- \E\{w(Z) g(Z)\} \biggr\rvert
< \frac{\epsilon}{2}\biggl(w_\infty \alpha
+ \frac{1}{n} \sum_{i=1}^{n} w(Z_i')g(Z_i')
+ \E\{w(Z)g(Z)\}\biggr)
\end{equation*}
imply
\begin{equation*}
\begin{split}
&\biggl\lvert \frac{1}{n} \sum_{i=1}^{n} w(Z_i)g(Z_i)
- \frac{1}{n} \sum_{i=1}^{n} w(Z_i')g(Z_i')\biggr\rvert\\
&\qquad\ge \frac{\epsilon}{2}\biggl(w_\infty \alpha
+ \frac{2}{n} \sum_{i=1}^{n} w(Z_i)g(Z_i)
- \frac{1}{n} \sum_{i=1}^{n} w(Z_i')g(Z_i') + \E\{w(Z)g(Z)\}\biggr).
\end{split}
\end{equation*}
Rearranging terms, we have
\begin{equation*}
\begin{split}
&\biggl\lvert \frac{1}{n} \sum_{i=1}^{n} w(Z_i)g(Z_i)
- \frac{1}{n} \sum_{i=1}^{n} w(Z_i')g(Z_i')\biggr\rvert
- \frac{3\epsilon}{4}\biggl(\frac{1}{n} \sum_{i=1}^{n} w(Z_i)g(Z_i)
- \frac{1}{n} \sum_{i=1}^{n} w(Z_i')g(Z_i')\biggr)\\
&\qquad> \frac{\epsilon}{2}  w_\infty \alpha
+ \frac{\epsilon}{4}\biggl(\frac{1}{n} \sum_{i=1}^{n} w(Z_i)g(Z_i)
+ \frac{1}{n} \sum_{i=1}^{n} w(Z_i')g(Z_i')\biggr)
+ \frac{\epsilon}{2}\E\{w(Z)g(Z)\}.
\end{split}
\end{equation*}
Together with $0 < 1 + \frac{3}{4}\epsilon < 2$ and $\E[w(Z)g(Z)] \ge0$,
this in turn implies
\begin{equation*}
\biggl\lvert \frac{1}{n} \sum_{i=1}^{n} w(Z_i)g(Z_i)
- \frac{1}{n} \sum_{i=1}^{n} w(Z_i')g(Z_i')\biggr\rvert
> \frac{\epsilon}{8} \biggl(2 w_\infty \alpha
+ \frac{1}{n} \sum_{i=1}^{n} w(Z_i)g(Z_i)
+ \frac{1}{n} \sum_{i=1}^{n} w(Z_i')g(Z_i')\biggr).
\end{equation*}

Let $g^* \in \cG$ be a function such that
\begin{equation*}
\biggl\lvert \frac{1}{n} \sum_{i=1}^{n} w(Z_i)g^*(Z_i)
- \E\{w(Z)g^*(Z)\} \biggr\rvert
> \epsilon\biggl(w_\infty \alpha + \frac{1}{n} \sum_{i=1}^{n} w(Z_i)g^*(Z_i)
+ \E\{w(Z)g^*(Z)\}\biggr)
\end{equation*}
if such a function exists, and let $f^*$ be any arbitrary function in $\cF$
if such a function does not exist.
Hence,
\begin{equation*}
\begin{split}
&\P\biggl\{\exists g \in \cG:
\biggl\lvert \frac{1}{n} \sum_{i=1}^{n} w(Z_i)g(Z_i)
- \frac{1}{n} \sum_{i=1}^{n} w(Z_i')g(Z_i')\biggr\rvert\\
&\phantom{\P\biggl\{\exists g \in \cG:}\qquad
> \frac{\epsilon}{8} \biggl(2 w_\infty \alpha
+ \frac{1}{n} \sum_{i=1}^{n} w(Z_i)g(Z_i)
+ \frac{1}{n} \sum_{i=1}^{n} w(Z_i')g(Z_i')\biggr)
\biggr\}\\
&\qquad> \P\biggl\{
\biggl\lvert \frac{1}{n} \sum_{i=1}^{n} w(Z_i)g^*(Z_i)
- \E\{w(Z) g^*(Z)\} \biggr\rvert\\
&\qquad\phantom{> \P\biggl\{}\qquad> \epsilon\biggl(w_\infty \alpha + \frac{1}{n} \sum_{i=1}^{n} w(Z_i)g^*(Z_i)
+ \E\{w(Z)g^*(Z)\}\biggr),\\
&\qquad\phantom{{}> \P\biggl\{}
\biggl\lvert \frac{1}{n} \sum_{i=1}^{n} w(Z_i')g^*(Z_i')
- \E\{w(Z)g^*(Z)\} \biggr\rvert\\
&\qquad\phantom{> \P\biggl\{}\qquad< \frac{\epsilon}{2}\biggl(
w_\infty \alpha
+ \frac{1}{n} \sum_{i=1}^{n} w(Z_i')g^*(Z_i')
+ \E\{w(Z)g^*(Z)\}\biggr)
\biggr\}\\
&\qquad= \E\biggl[
\II\biggl\{
\biggl\lvert \frac{1}{n} \sum_{i=1}^{n} w(Z_i)g^*(Z_i)
- \E\{w(Z)g^*(Z)\} \biggr\rvert\\
&\qquad\phantom{= \E\biggl[}\qquad
> \epsilon\biggl(w_\infty \alpha + \frac{1}{n} \sum_{i=1}^{n} w(Z_i)g^*(Z_i)
+ \E\{w(Z)g^*(Z)\}\biggr)
\biggr\}\\
&\phantom{\qquad= \E\biggl[}
\times \P\biggl\{\biggl\lvert \frac{1}{n} \sum_{i=1}^{n} w(Z_i')g^*(Z_i')
- \E\{w(Z)g^*(Z)\} \biggr\rvert\\
&\phantom{\qquad= \E\biggl[\times \P\biggl\{}
< \frac{\epsilon}{2}\biggl(w_\infty \alpha
+ \frac{1}{n} \sum_{i=1}^{n} w(Z_i')g^*(Z_i')
+ \E\{w(Z)g^*(Z)\}\biggr) \biggm\vert (Z_i)_{i=1}^n\biggr\}
\biggr].
\end{split}
\end{equation*}
Since $0 \le w(Z_i') g^*(Z_i') \le w_\infty B \quad (i =1, \ldots, n)$,
by \citep[Lemma~11.2]{gyorfi2002distribution}, one gets
\begin{equation*}
\begin{split}
&\P\biggl\{\biggl\lvert \frac{1}{n} \sum_{i=1}^{n} w(Z_i')g^*(Z_i')
- \E[w(Z)g^*(Z)] \biggr\rvert\\
&\phantom{\P\biggl\{}\qquad> \frac{\epsilon}{2}\biggl(
w_\infty \alpha
+ \frac{1}{n} \sum_{i=1}^{n} w(Z_i')g^*(Z_i')
+ \E[w(Z)g^*(Z)]\biggr) \biggm\vert (Z_i)_{i=1}^n\biggr\}\\
&\qquad\le \frac{w_\infty B}{4 (\epsilon/2)^2 \alpha w_\infty n}
= \frac{B}{\epsilon^2 \alpha n}.
\end{split}
\end{equation*}
Therefore, for $n \ge \frac{2B}{\epsilon^2 \alpha}$,
the probability inside the expectation is greater than $\frac{1}{2}$.
Hence,
\begin{equation*}
\begin{split}
&\P\biggl\{\exists g \in \cG:
\biggl\lvert \frac{1}{n} \sum_{i=1}^{n} w(Z_i')g(Z_i')
- \frac{1}{n} \sum_{i=1}^{n} w(Z_i')g(Z_i')\biggr\rvert\\
&\phantom{\P\biggl\{\exists g \in \cG}\qquad
> \frac{\epsilon}{8} \biggl(2 w_\infty \alpha
+ \frac{1}{n} \sum_{i=1}^{n} w(Z_i)g(Z_i)
+ \frac{1}{n} \sum_{i=1}^{n} w(Z_i')g(Z_i')\biggr)
\biggr\}\\
&\qquad\ge \frac{1}{2} \P\biggl\{
\biggl\lvert \frac{1}{n} \sum_{i=1}^{n} w(Z_i)g^*(Z_i)
- \E\{w(Z)g^*(Z)\} \biggr\rvert\\
&\qquad\phantom{\ge \frac{1}{2} \P\biggl\{}\qquad
> \epsilon\biggl(w_\infty \alpha + \frac{1}{n} \sum_{i=1}^{n} w(Z_i)g^*(Z_i)
+ \E\{w(Z)g^*(Z)\}\biggr)
\biggr\}\\
&\qquad= \frac{1}{2}\P\biggl\{\exists g \in \cG:
\biggl\lvert \frac{1}{n} \sum_{i=1}^{n} w(Z_i)g(Z_i)
- \E\{w(Z)g^*(Z)\} \biggr\rvert\\
&\qquad\phantom{= \frac{1}{2}\P\biggl\{\exists g \in \cG:}\qquad
> \epsilon\biggl(w_\infty \alpha + \frac{1}{n} \sum_{i=1}^{n} w(Z_i)g(Z_i)
+ \E\{w(Z)g(Z)\}\biggr)
\biggr\}.
\end{split}
\end{equation*}

\noindent\textsc{Step 2.} Introducing random signs.

Let $U_1, \ldots, U_n$ be independent Rademacher random variables,
i.e., uniform distributed on $\{-1, 1\}$,
and independent of $Z_1, \ldots, Z_n$ and $Z_1', \ldots, Z_n'$.
Since $Z_i$ and $Z_i'$ are i.i.d. and independent of everything else,
interchanging them independently does not affect the previous probability.
Hence, we have that
\begin{equation*}
\begin{split}
&\P\biggl\{
\biggl\lvert \frac{1}{n} \sum_{i=1}^{n} w(Z_i)g(Z_i)
- \frac{1}{n} \sum_{i=1}^{n} w(Z_i')g(Z_i')\biggr\rvert\\
&\phantom{\P\biggl\{}\qquad> \frac{\epsilon}{8} \biggl(2 w_\infty \alpha
+ \frac{1}{n} \sum_{i=1}^{n} w(Z_i)g(Z_i)
+ \frac{1}{n} \sum_{i=1}^{n} w(Z_i')g(Z_i')\biggr)
\biggr\}\\
 &\qquad=\P\biggl\{
\biggl\lvert \frac{1}{n} \sum_{i=1}^{n} U_i w(Z_i)g(Z_i)
- \frac{1}{n} \sum_{i=1}^{n} U_i w(Z_i')g(Z_i')\biggr\rvert\\
&\phantom{\qquad=\P\biggl\{}> \frac{\epsilon}{8} \biggl(2 w_\infty \alpha
+ \frac{1}{n} \sum_{i=1}^{n} w(Z_i)g(Z_i)
+ \frac{1}{n} \sum_{i=1}^{n} w(Z_i')g(Z_i')\biggr)
\biggr\}\\
&\qquad\le 2\P\biggl\{
\biggl\lvert \frac{1}{n} \sum_{i=1}^{n} U_i w(Z_i)g(Z_i)\biggr\rvert
> \frac{\epsilon}{8} \biggl(w_\infty \alpha
+ \frac{1}{n} \sum_{i=1}^{n} w(Z_i)g(Z_i)
\biggr)\biggr\}.
\end{split}
\end{equation*}

\noindent\textsc{Step 3.} Conditioning and covering.

We condition on $(Z_i)_{i=1}^n$,
which is equivalent to fixing $(z_i)_{i=1}^n$,
and consider
\begin{equation*}
\P\biggl\{\exists g \in \cG:
\biggl\lvert \frac{1}{n} \sum_{i=1}^{n} U_i w(z_i)g(z_i)\biggr\rvert
> \frac{\epsilon}{8} \biggl(w_\infty \alpha
+ \frac{1}{n} \sum_{i=1}^{n} w(z_i)g(z_i)
\biggr)
\biggr\}.
\end{equation*}
Let $\delta > 0$ and $\cG^\delta$ be an $L_\infty$-cover of $\cG$.
Then, for any $g \in \cG$, there exists a $\gt \in \cG_\delta$ such that
\begin{equation*}
\frac{1}{n}\sum_{i=1}^{n} w(z_i) \abs{g(z_i) - \gt(z_i)}
\le \frac{1}{n} \sum_{i=1}^{n} w(z_i) \delta
\le w_\infty \delta.
\end{equation*}
Therefore,
\begin{equation*}
\begin{split}
\biggl\lvert \frac{1}{n} \sum_{i=1}^{n} U_i w(z_i)g(z_i)\biggr\rvert
&\le \biggl\lvert \frac{1}{n} \sum_{i=1}^{n} U_i w(z_i)\gt(z_i)\biggr\rvert
+ \biggl\lvert \frac{1}{n} \sum_{i=1}^{n}
U_i w(z_i)(g(z_i) - \gt(z_i))\biggr\rvert\\
&\le \biggl\lvert \frac{1}{n} \sum_{i=1}^{n} U_i w(z_i)\gt(z_i)\biggr\rvert
+ \frac{1}{n} \sum_{i=1}^{n} w(z_i) \abs{g(z_i) - \gt(z_i)}\\
&\le \biggl\lvert \frac{1}{n} \sum_{i=1}^{n} U_i w(z_i)\gt(z_i)\biggr\rvert
+ w_\infty \delta
\end{split}
\end{equation*}
and
\begin{equation*}
\begin{split}
\frac{1}{n} \sum_{i=1}^{n} w(z_i)g(z_i)
&\ge \frac{1}{n} \sum_{i=1}^{n} w(z_i)\gt(z_i)
- \biggl\lvert\frac{1}{n} \sum_{i=1}^{n} w(z_i)(g(z_i) - \gt(z_i))\biggr\rvert\\
&\ge \frac{1}{n} \sum_{i=1}^{n} w(z_i)\gt(z_i)
- \frac{1}{n} \sum_{i=1}^{n} w(z_i)\abs{g(z_i) - \gt(z_i)}\\
&\ge \frac{1}{n} \sum_{i=1}^{n} w(z_i)\gt(z_i) - w_\infty \delta.
\end{split} 
\end{equation*}
Using these and a union bound gives
\begin{equation*}
\begin{split}
&\P\biggl\{\exists g \in \cG:
\biggl\lvert \frac{1}{n} \sum_{i=1}^{n} U_i w(z_i)g(z_i)\biggr\rvert
> \frac{\epsilon}{8} \biggl(w_\infty \alpha
+ \frac{1}{n} \sum_{i=1}^{n} w(z_i)g(z_i)\biggr)
\biggr\}\\
&\qquad\le \abs{G_\delta} \max_{g \in \cG_\delta}
\P\biggl\{\biggl\lvert \frac{1}{n} \sum_{i=1}^{n} U_i w(z_i)g(z_i)\biggr\rvert
> \frac{\epsilon}{8} \biggl(w_\infty \alpha
+ \frac{1}{n} \sum_{i=1}^{n} w(z_i)g(z_i)
- w_\infty \delta\biggr) - w_\infty \delta
\biggr\}.
\end{split}
\end{equation*}
By choosing $\delta = \frac{\epsilon \alpha}{17}$,
we have
\begin{equation*}
\frac{\epsilon \alpha}{8} - \frac{\epsilon \delta}{8} - \delta
= \frac{\epsilon \alpha}{8} - \frac{\epsilon^2\alpha}{136}
- \frac{\epsilon \alpha}{17}
\ge \frac{\epsilon \alpha}{17}.
\end{equation*}
Therefore, we have
\begin{equation*}
\begin{split}
&\P\biggl\{\exists g \in \cG:
\biggl\lvert \frac{1}{n} \sum_{i=1}^{n} U_i w(z_i)g(z_i)\biggr\rvert
> \frac{\epsilon}{8} \biggl(w_\infty \alpha
+ \frac{1}{n} \sum_{i=1}^{n} w(z_i)g(z_i)\biggr)
\biggr\}\\
&\qquad\le N_\infty\biggl(\frac{\epsilon\alpha}{17}, \cG, n\biggr)
\max_{g \in \cG^{\frac{\epsilon\alpha}{17}}}
\P\biggl\{\biggl\lvert \frac{1}{n} \sum_{i=1}^{n} U_i w(z_i)g(z_i)\biggr\rvert
> \frac{\epsilon}{17}w_\infty \alpha
+ \frac{\epsilon}{8}\cdot\frac{1}{n} \sum_{i=1}^{n} w(z_i)g(z_i)
\biggr\}.
\end{split}
\end{equation*}

\noindent\textsc{Step 4.} Hoeffding's inequality.

Fix $z_1, \ldots, z_n$, we wish to bound
\begin{equation*}
\P\biggl\{
\biggl\lvert \frac{1}{n} \sum_{i=1}^{n} U_i w(z_i)g(z_i)\biggr\rvert
> \frac{\epsilon}{17} w_\infty \alpha
+ \frac{\epsilon}{8} \cdot \frac{1}{n} \sum_{i=1}^{n} w(z_i)g(z_i)
\biggr\}.
\end{equation*}
Since $\abs{U_i w(z_i)g(z_i)} \le w(z_i)g(z_i)$ almost surely,
by Hoeffding's inequality (see, e.g.,
\citep[Theorem~2.8]{boucheron2003concentration}),
we have
\begin{equation*}
\begin{split}
&\P\biggl\{
\biggl\lvert \frac{1}{n} \sum_{i=1}^{n} U_i w(z_i)g(z_i)\biggr\rvert
> \frac{\epsilon}{17} w_\infty \alpha
+ \frac{\epsilon}{8} \cdot \frac{1}{n} \sum_{i=1}^{n} w(z_i)g(z_i)
\biggr\}\\
&\qquad\le 2\exp\biggl(-\frac{\epsilon^2 \bigl(\frac{n}{17} w_\infty \alpha
+ \frac{1}{8}\sum_{i=1}^{n} w(z_i)g(z_i)\bigr)^2}
{2 \sum_{i=1}^n w(z_i)^2 g(z_i)^2}\biggr)\\
&\qquad\stackrel{(a)}{\le} 
2\exp\biggl(-\frac{\epsilon^2 \bigl(\frac{n}{17} w_\infty \alpha
+ \frac{1}{8}\sum_{i=1}^{n} w(z_i)g(z_i)\bigr)^2}
{2 w_\infty B\sum_{i=1}^n w(z_i) g(z_i)}\biggr)
\stackrel{(b)}{\le} 2\exp\biggl(-\frac{n \epsilon^2 \alpha}{68 B}\biggr)
\end{split}
\end{equation*}
where $(a)$ is due to the definition of $w_\infty$ and $g$
and $(b)$ is by using $a + b \ge 2\sqrt{ab}$ for any $a, b \ge 0$.

The lemma is hence proved by combining the previous steps.

\section{Proof of Lemma~\ref{lem:composition}}
\begin{proof}[Proof of Lemma~\ref{lem:composition}]
By definition of the VC-dimension (see, e.g., \citep[Section 3.3]{anthony1999neural}),
the growth function
\begin{equation*}
\Pi_{\cH}(m) \le \max_{\sum_{k=1}^{K}m_k = m} \prod_{k=1}^{K} \Pi_{\cH_k}(m_k).
\end{equation*}
Let $K_0 \coloneq \{k: m_k \le p_k\}$ and $K_1 = [K] \setminus K_0$
be its complement.
Denote $p_k = \VCdim(\cH_k)$ and $p = \sum_{k=1}^{K} p_k$.
By \citep[Theorem~3.7]{anthony1999neural},
\begin{equation*}
\begin{split}
\prod_{k=1}^{K} \Pi_{\cH_k}(m_k)
&= \prod_{k \in K_0} \Pi_{\cH_k}(m_k) \prod_{k \in K_1} \Pi_{\cH_k}(m_k)
\le 2^{\sum_{k \in K_0} m_k}
\prod_{k \in K_1} \biggl(\frac{em_k}{p_k}\biggr)^{p_k}\\
&= (2^{m_0/p_0})^{p_0}
\prod_{k \in K_1} \biggl(\frac{em_k}{p_k}\biggr)^{p_k}
\end{split}
\end{equation*}
where $m_0 \coloneq \sum_{k \in K_0} m_k$ and
$p_0 \coloneq \sum_{k \in K_0} p_k$.
The weighted AM--GM inequality gives
\begin{equation*}
\begin{split}
(2^{m_0/p_0})^{p_0} \prod_{k \in K_1} \biggl(\frac{em_k}{p_k}\biggr)^{p_k}
&\le \biggl(\frac{p_0 2^{m_0/p_0}
+ \sum_{k \in K_1} e m_k}{p_0 + \sum_{k \in K_1} p_k}\biggr)
^{p_0 + \sum_{k \in K_1} p_k}\\
&= \biggl(\frac{p_0 2^{m_0/p_0}
+ \sum_{k \in K_1} e m_k}{p}\biggr)^{p}\\
&\stackrel{(a)}{\le}
\biggl(\frac{p_0 + m_0 + \sum_{k \in K_1} e m_k}{p}\biggr)^{p}
\le \biggl(\frac{p_0 + e m}{p}\biggr)^{p}\\
&\stackrel{(b)}{\le}\biggl(\frac{(1+e) m}{p}\biggr)^{p}
\end{split}
\end{equation*}
where we used $2^x \le 1 + x$ for $0 \le x \le 1$ in $(a)$
and $m \ge p$ in $(b)$.
Therefore, we arrive at
\begin{equation*}
\Pi_{\cH}(m) \le \biggl(\frac{(1+e) m}{p}\biggr)^{p}.
\end{equation*}
Hence, in order for $2^m > \Pi_{\cH}(m)$, it suffices to take $m = 4p$.
The claim directly follows.
\end{proof}

\section{Proof of Lemma~\ref{lem:Pdim}}
We first show an upper bound on the VC-dimension of the neural knowledge graph function class.
Since the pseudo-dimension of a neural network
is identical to the VC-dimension of an equivalent network with an additional
scalar parameter to the output unit,
Theorem~\ref{lem:Pdim} is directly implied by the following result.
Note that here we show a more general result with precise constants for neural networks
with piecewise polynomial activation function,
which directly implies these with ReLU activation function.
For simplicity of presentation,
we also ignore the bias parameters which can be absorbed into the weight matrices
by creating additional hidden units with value $1$.

\begin{lemma}[VC-dimension] \label{lem:VC}
Let  $\cF$ be the nerual knowledge graph function class in Definition~\ref{def:ipkg} and~\ref{def:cnkg}
with piecewise polynomial activation function.
Denote $W \coloneq \frac{1}{K}\sum_{k=1}^{K} W_k$ the total number of parameters
for each relation type except for the embedding layer,
which has $N D$ parameters.
Then, the VC-dimension of $\cF$ satisfies
\begin{equation*}
\VCdim(\cF) \le 3(LND + \Lb K W) \log (8eU)
\end{equation*}
where $\Lb \coloneq \frac{1}{K W} \sum_{\ell=1}^{L} \sum_{i=1}^{\ell} W_k^{(i)}$.
For the inner product model in Definition~\ref{def:ipkg},
$U = 2S\sum_{\ell=1}^{L} (\ell+1)H_\ell$ if $Q = 1$
and $U = 6S\sum_{\ell=1}^{L}H_\ell Q^\ell$ if $Q \ge 2$.
For the concatenation model in Definition~\ref{def:ipkg},
$U = S\sum_{\ell=1}^{L} (\ell+1)H_\ell$ if $Q = 1$
and $U = 3S\sum_{\ell=1}^{L}H_\ell Q^\ell$ if $Q \ge 2$.
\end{lemma}

The following lemma is a simple fact by the structure of the function class,
which becomes handy in several places.

\begin{lemma}
\label{lem:fixed}
Let $\cF$ be the neural knowledge graph function class
with VC-dimension $\VCdim(\cH)$.
For a fixed function $g: \cX \to \RR$,
define the function class $\cFt \coloneq \{\ft:
\ft = f - g, \forall f \in \cF\}$.
Then the VC-dimension of $\cFt$, $\VCdim(\cFt) = \VCdim(\cF)$.
\end{lemma}

\begin{proof}
In the proof of Lemma~\ref{lem:VC},
define $\ft(x_m, a) \coloneq f(x_m, a) + g(x_m)$.
The only change to the function is the activation $\eta_L$,
which is added by a constant in each $f(x_m, a)$.
Since a piecewise polynomial plus a constant is still
a piecewise polynomial with the same number of pieces and degree,
by following the proof,
the claim directly holds.
\end{proof}

By definition of the pseudo-dimension (see, e.g.,
\citep[Definition~11.2]{anthony1999neural}),
it is equivalent to the VC-dimension subtracting a fixed function.
Using Lemma~\ref{lem:fixed}, we immediately have
the following pseudo-dimension bound for the neural network function class.

For the knowledge graph model under consideration,
we have the embedding parameters $W_0 = ND$ shared among all $K$ relations
and the total neural network parameters is $\sum_{k=1}^{K} W_k$.
Since the ReLU is a piecewise polynomial with $S = 2$ and $Q = 1$,
we have $U \le 4 \sum_{\ell=1}^{L-1}(\ell+1)H^{(\ell)}
\le 4L \sum_{\ell=1}^{L-1} H^{(\ell)}
\le 4 L \sum_{k=1}^{K} W_k \le 4 L K W$
where we used that the total number of hidden units is smaller than
the total number of feed-forward network weight parameters.
In addition, we also have $\Lb \le L$.
The lemma is hence proved.

\section{Proof of Lemma~\ref{lem:VC}}
The proof of Lemma~\ref{lem:VC} is similar to that of
\citep[Theorem~1]{bartlett1998almost}
with modifications to cope with the embedding layer
and the hypothesis space $\cH$.
We also use the improvements made in \citep[Theorem~7]{bartlett2019nearly}
together with some further improvements on counting the polynomials.
We first show the proof for the model in Definition~\ref{def:cnkg} in detail
and then highlight the differences for the model in Definition~\ref{def:ipkg}.
We begin by stating the following key lemma due to \citet{bartlett1998almost}.
\begin{lemma}[{\citep[Lemma~1]{bartlett1998almost}}] \label{lem:poly}
Suppose $f_1(\cdot), \ldots, f_m(\cdot)$ are fixed polynomials of degree at most
$d$ in $n \le m $ variables.
Then the number of distinct sign vectors $(\sgn(f_1(a)), \dots, \sgn(f_m(a)))$
that can be generated by varying $a \in \RR^n$ is at most $2(2emd/n)^n$.
\end{lemma}

\begin{proof}[Proof of Lemma~\ref{lem:VC}]
Since by definition of the VC-dimension~\citep[Chapter~3.3]{anthony1999neural},
adding a monotone function to the output does not increase the VC-dimension of the function class.
We hence focus on the function class without $\rho$ and prove an upper bound on the VC-dimension.
Let us begin by considering the concatenation model in Definition~\ref{def:cnkg}.

We first combine the feed-forward networks for all relations into one big feed-forward network
by stacking their hidden units in each layer.
Let $W^{(\ell)} \coloneq \sum_{k=1}^K W_k^{(\ell)}$ be the number of total weights
of the combined feed-forward network.
We also define $\wt{W}^{(\ell)} \coloneq \sum_{i=0}^{\ell} W^{(i)}$ 
the number of the parameters \emph{up to} layer $\ell$
(including the embedding).
Hence $\wt{W}^{(L)}$ is the number of total parameters in the neural network.
We use $H^{(\ell)} \coloneq \sum_{k=1}^K H_k^{(\ell)}, \ell = 1, \ldots, L$
to represent the total number of hidden units in each layer
and $H^{(0)} = D$ to denote the embedding dimension.

For an input $x \in \cX$ and parameters $a \in \RR^{\wt{W}^{(L)}}$,
let $f(x, a)$ denote the output of the neural network.
Given $x_1, \ldots, x_m \in \cX$ where $m = \abs{\cX}$,
we wish to bound
\begin{equation*}
T \coloneq \lvert \{ (\sgn(f(x_1, a)), \ldots, \sgn(f(x_m, a))):
a \in \RR^W \} \rvert.
\end{equation*}
For any partition $\cS = \{P_1, \ldots, P_S\}$ of the parameter space $\RR^W$,
we have
\begin{equation*}
T \le \sum_{i=1}^{\abs{\cS}}
\lvert \{ (\sgn(f(x_1, a)), \ldots, \sgn(f(x_m, a))): a \in P_i \} \rvert.
\end{equation*}
We next construct a sequence of partitions $\cS_0, \cS_1, \ldots, S_L$
that are successive refinements.
They are built from the embedding layer to the output layer recursively
by fixing the weights in later layers.
For each element $A \in \cS_\ell$,
the output from the neural network is a polynomial function of the parameters
up to layer $\ell$.
In the embedding layer,
since there is no activation involved,
the hidden units is a degree one polynomial of the embedding parameters.
Hence the partition is just $\RR^{W^{(0)}}$, i.e., $\abs{\cS_0} = 1$.
Define the tuple of functions
\begin{equation*}
(\sgn(h_{i,j}(a) - t_s)),\qquad i \in [m], j \in [H^{(\ell)}], s \in [S],
\end{equation*}
where $h_{i,j}(a)$ is the value of the $j$th hidden unit for sample $x_i$
before the activation function.
For each $A \in \cS_{\ell-1}$,
by definition $h_{i,j}(a)$ is a polynomial function of parameters up to $\ell - 1$.
Therefore, it is also a polynomial function parameters up to $\ell$.
Each value of the tuple determines a region where the hidden unit values is still
a polynomial function of parameters up to $\ell - 1$.
The number of all possible values of such tuples can be obtained
by Lemma~\ref{lem:poly}
and we further refine $A \in \cS_{\ell-1}$ by partitioning it into these regions.
Therefore, we have
\begin{equation*}
\abs{\cS_{\ell}}
\le 2\biggl(\frac{2emH_{\ell} S q_{\ell}}{\wt{W}_{\ell}}\biggr)^{\wt{W}_{\ell}}
\abs{\cS_{\ell-1}},
\end{equation*}
where $q_{\ell}$ is the maximum degree of parameters achieved at the $\ell$th
layer and satisfies the recursion
\begin{equation*}
q_{\ell} = Q (q_{\ell-1} + 1), \quad q_0 = 1
\end{equation*}
for $\ell = 1, \ldots, L-1$.
Solving the recursion, we obtain that $q_\ell = \ell+1$ for $Q = 1$ and
\begin{equation*}
q_{\ell} = Q^{\ell}\biggl(1 + \frac{Q}{Q-1}\biggr) - \frac{Q}{Q-1}
\le 3Q^{\ell} \text{ for $Q \ge 2$} 
\end{equation*}
where we used $\frac{a}{a-1} \le \frac{b}{b-1}$
for $a \ge b > 1$.
Applying the recursion iteratively gives
\begin{equation*}
\abs{\cS_{L-1}} \le \prod_{\ell=1}^{L-1}
2\biggl(\frac{2emH^{(\ell)}Sq_{\ell}}{\wt{W}^{(\ell)}}\biggr)^{\wt{W}^{(\ell)}}.
\end{equation*}
By using Lemma~\ref{lem:poly} again, we have for each $P \in \cS_{L-1}$,
\begin{equation*}
\lvert \{ (\sgn(f(x_1, a)), \ldots, \sgn(f(x_m, a))): a \in P \} \rvert
\le 2 \biggl(\frac{2em(q_{L-1} + 1)}{\wt{W}^{(L)}}\biggr)^{\wt{W}^{(L)}}.
\end{equation*}
Combining the above two displays and denoting $H_L = 1, q_L = q_{L-1} + 1$,
\begin{equation*}
T \le 2\biggl(\frac{2em q_L}{\wt{W}^{(L)}}\biggr)^{\wt{W}^{(L)}}
\abs{\cS_{L-1}}
\le \prod_{\ell=1}^{L}
2\biggl(\frac{2emH^{(\ell)} S q_\ell}{\wt{W}^{(\ell)}}\biggr)^{\wt{W}^{(\ell)}}
\le 2^L \biggl(\frac{2emS\sum_{\ell=1}^L H^{(\ell)} q_\ell}
{\sum_{\ell=1}^{L}\wt{W}^{(\ell)}}\biggr)^{\sum_{\ell=1}^{L}\wt{W}^{(\ell)}}
\end{equation*}
where the last inequality is by a weighted AM--GM.
Hence, the VC-dimension is upper bounded by $m$ such that
\begin{equation*}
2^m \le 2^L \biggl(\frac{2em U}
{\sum_{\ell=1}^{L}\wt{W}^{(\ell)}}\biggr)^{\sum_{\ell=1}^{L}\wt{W}^{(\ell)}}
\le \biggl(\frac{4em U}
{\sum_{\ell=1}^{L}\wt{W}^{(\ell)}}\biggr)^{\sum_{\ell=1}^{L}\wt{W}^{(\ell)}}
\end{equation*}
where we denoted $U = S\sum_{\ell=1}^L H^{(\ell)} q_\ell$.
And by the recursion of $q_\ell$, we know that
$U = S \sum_{\ell=1}^L (\ell+1) H_{\ell}$ if $Q = 1$
and $U = 3S \sum_{\ell=1}^L H_{\ell} Q^\ell$ if $Q \ge 2$.
By \citep[Lemma~18]{bartlett2019nearly}, since $U \ge 2$ implies $4eU \ge 16$,
\begin{equation*}
\VCdim(\cH) \le \biggl(\sum_{\ell=0}^{L}\wt{W}_\ell\biggr)
\biggl(\log_2 (8eU) + \log_2 \log_2 (4eU)\biggr)
\le 3 (LND + \Lb K W) \log (8eU)
\end{equation*}
where $\Lb \coloneq \frac{1}{K W} \sum_{\ell=1}^{L} \sum_{i=1}^\ell W^{(i)}
= \frac{1}{K W} \sum_{k=1}^K \sum_{\ell=1}^{L} \sum_{i=1}^\ell W_k^{(i)}$.

We now turn to the inner product model in Definition~\ref{def:ipkg}.
The proof follows the same steps, except for the calculation of polynomial degrees.
Here the output is a product of functions of previous weights.
This does not change the size of the partitions but the degrees of the polynomials.
In each layer, the polynomial degree becomes $q_\ell' = 2 q_\ell, \ell = 0, \ldots, L$.
The claim hence directly follows.
\end{proof}

\end{document}